\documentclass[a4paper,10pt]{article}
\usepackage{amsmath,amssymb,amstext,amsfonts}
\usepackage{amsthm, mathrsfs}
\usepackage{graphicx,subfigure}
\usepackage{url}
\usepackage{authblk}
\usepackage{natbib}
\usepackage{dsfont} 

 \newcommand{\sgn}{\operatorname{sgn}}

\newcommand{\ud}{\, \mathrm{d}}

\newcommand{\indicatorfn}{\mathds{1}}

\newcommand{\Omegaset}{\Omega}
\newcommand{\timehorizon}{T}
\newcommand{\timeset}{[0,\timehorizon]}

\newcommand{\filtration}{\mathcal{F}}

\newcommand{\Pmeasure}{\mathbb{P}}
\newcommand{\Qmeasure}{\mathbb{Q}}

\newcommand{\expectation}{\textrm{E}}

\newcommand{\expectationQ}{\expectation^{\Qmeasure}}

\newcommand{\measurablespace}{\Omegaset,\filtration}

\newcommand{\probspace}{(\Omegaset,\filtration,\Pmeasure)}

\newcommand{\filtprobspace}{(\Omegaset,\filtration,\Pmeasure, \{ \filtration_{t} \})}

\newcommand{\BMdim}{N}
\newcommand{\BM}{W}
\newcommand{\vBM}{\mathbf{W}}

\newcommand{\markovspace}{I}
\newcommand{\Markovgenerator}{G}
\newcommand{\markovgenerator}{g}

\newcommand{\markovchain}{\alpha}
\newcommand{\markovdim}{D}
\newcommand{\markovinitialstate}{i_{0}}

\newcommand{\Pmarkovgenerator}{\markovgenerator}

\newcommand{\markovmartingale}{M}

\newcommand{\markovmartingaleij}{\markovmartingale_{ij}}
\newcommand{\markovsquareqvprocess}{N}
\newcommand{\markovcompensator}{\lambda}

\newcommand{\realnumbers}{\mathbb{R}}
\newcommand{\realnumbersN}{\realnumbers^{\BMdim}}

\newcommand{\bankaccount}{S_{0}}
\newcommand{\stockprice}{S}
\newcommand{\vstockprice}{\mathbf{\stockprice}}
\newcommand{\sharperatio}{(SR)}
\newcommand{\geninstrumentprice}{\Pi}
\newcommand{\geninstrumentmeanret}{\mu^{\geninstrumentprice}}
\newcommand{\geninstrumentvolBM}{\sigma^{\geninstrumentprice}}
\newcommand{\vgeninstrumentvolBM}{\pmb{\sigma}^{\geninstrumentprice}}
\newcommand{\geninstrumentvolMC}{\gamma^{\geninstrumentprice}}
\newcommand{\vgeninstrumentvolMC}{\pmb{\gamma}^{\geninstrumentprice}}

\newcommand{\geninstrumentvolatility}{\nu}
\newcommand{\meanret}{b}
\newcommand{\vmeanret}{\mathbf{\meanret}}
\newcommand{\vol}{\sigma}
\newcommand{\vvol}{\pmb{\vol}}
\newcommand{\vvolT}{\vvol^{\top}}

\newcommand{\RNderivative}{L}
\newcommand{\RNforBM}{h}
\newcommand{\vRNforBM}{\mathbf{\RNforBM}}
\newcommand{\vRNforBMT}{\mathbf{\RNforBM}^{\top}}
\newcommand{\RNforMC}{\eta}
\newcommand{\vRNforMC}{\pmb{\eta}}

\newcommand{\RNforMCopt}{\bar{\RNforMC}}
\newcommand{\vRNforMCopt}{\bar{\vRNforMC}}
\newcommand{\Qmarkovgeneratoropt}{\RNforMCopt}

\newcommand{\vPBM}{\vBM}

\newcommand{\vQBM}{\vBM^{\Qmeasure}}

\newcommand{\Pmarkovmartingaleij}{\markovmartingale_{ij}}

\newcommand{\Qmarkovmartingale}{\markovmartingale^{\Qmeasure}}

\newcommand{\RNforMCmmm}{\RNforMC^{\textrm{min}}}

\newcommand{\vRNforBMmmm}{\vRNforBM^{\textrm{min}}}
\newcommand{\vRNforMCmmm}{\vRNforMC^{\textrm{min}}}
\newcommand{\Qmeasuremmm}{\Qmeasure^{\textrm{min}}}

\newcommand{\hilbertspace}{\mathcal{H}}

\newcommand{\contingentclaim}{Z}
\newcommand{\contingentfn}{\Phi}
\newcommand{\gooddealbnd}{B}

\newcommand{\infinitoperator}{\mathbb{A}}
\newcommand{\infinitoperatorBMMC}{\infinitoperator^{(\vRNforBM, \vRNforMC)}}
\newcommand{\infinitoperatorMC}{\infinitoperatorBMMC}
\newcommand{\infinitoperatorMCopt}{\infinitoperator^{(\vRNforBM, \vRNforMCopt)}}
\newcommand{\infinitoperatormmm}{\infinitoperator^{(\vRNforBMmmm, \vRNforMCmmm)}}
\newcommand{\valuefn}{V}
\newcommand{\valuefnupper}{\valuefn}
\newcommand{\valuefnlower}{\valuefn}
\newcommand{\valuefnupperii}{\valuefn^{\textrm{upper}}}
\newcommand{\valuefnlowerii}{\valuefn^{\textrm{lower}}}
\newcommand{\valuefnmmm}{V^{\textrm{min}}}

\newcommand{\strikeprice}{K}
\newcommand{\genconstraint}{\tilde{B}}

\newcommand{\HJ}{Hansen-Jagannathan }
\newcommand{\gd}{good-deal }

\theoremstyle{plain}
\newtheorem{thm}{Theorem}[section]

\newtheorem{lma}[thm]{Lemma}
\newtheorem{propn}[thm]{Proposition}

\theoremstyle{definition}
\newtheorem{cond}[thm]{Condition}
\newtheorem{defn}[thm]{Definition}

\newtheorem{assumption}[thm]{Assumption}

\theoremstyle{remark}

\newtheorem{rmk}[thm]{Remark}

\numberwithin{equation}{section}

\title{Good-deal bounds in a regime-switching diffusion market}

\author{Catherine Donnelly}

\affil{RiskLab, ETH Zurich, Switzerland}

\begin{document}

\maketitle

\begin{abstract}
We consider option pricing in a regime-switching diffusion market.  As the market is incomplete, there is no unique price for a derivative.  We apply the \gd pricing bounds idea to obtain ranges for the price of a derivative.  As an illustration, we calculate the \gd pricing bounds for a European call option and we also examine the stability of these bounds when we change the generator of the Markov chain which drives the regime-switching.  We find that the pricing bounds depend strongly on the choice of the generator.
\end{abstract}

\section{Introduction} \label{intro}
Regime-switching market models are a way of capturing discrete shifts in market behavior.  These shifts could be due to a variety of reasons, such as changes in market regulations, government policies or investor sentiment.  First introduced by \citet{hamilton89.article}, regime-switching models have been shown in various empirical studies to be better at capturing market behavior than their non-regime-switching counterparts (for example, see \citet{angbakaert02.article}, \citet{gray96.article} and \citet{klaassen02.article}).

An example of regime-switching market is one in which there are only two regimes: a bear market regime and a bull market regime.  Suppose the market starts in a bull market regime, in which prices are generally rising.  It stays in this regime for a random length of time before switching to a bear market regime, in which prices are generally falling.  It then stays in the bear market for another random length of time before switching back to the bull market.  This cycle continues ad infinitum.

Due to the regime-switching, the market is incomplete and hence there is no unique risk-neutral martingale measure to use for pricing derivatives.  In fact, there are infinitely many possible risk-neutral martingale measures.  This means that not only is there no unique price for derivatives, but the range of prices which can be obtained from all the possible risk-neutral martingale measures are too wide to be useful in practice.

As prices of derivatives are not unique in incomplete markets, various suggestions have been made either on how to choose a single price or on how to obtain a more restricted, and therefore potentially more useful, range of prices.  We focus in this paper on the latter because it is the market which ultimately decides which risk-neutral martingale measure is used for pricing a derivative and so we should take into account our uncertainty about what the market price will be.  Therefore, we believe it is better to find a range of prices that the market-determined price might reasonably be expected to lie in, rather than determining a single price.

The idea that we build upon is that of the \emph{\gd bound}.  This idea is due to \citet{cochranesaa00.article} and is based on the Sharpe Ratio, which is the excess return on an investment per unit of risk.  Their idea is to bound the Sharpe Ratios of all possible assets in the market and thus exclude Sharpe Ratios which are considered to be too large.  The bound is called a \gd bound.  The method of applying the bound gives a set of risk-neutral martingale measures which can be used to price options.  This results in an \emph{upper and lower \gd bound} on the prices of an option.  The idea was streamlined and extended to models with jumps in \citet{bjorkslinko06.article}, and it is their approach that we follow in this paper.

The uncertainty measured by the size of the \gd bounds reflects the uncertainty within the market model concerning the price of the derivative.  It does not measure uncertainty concerning the choice of the model, by which we mean both the model structure - in this case a regime-switching model - and the model parameters.  Indeed, as we see concretely in a numerical example, changing the model parameters changes a derivative's pricing bounds for a fixed \gd bound.  Thus the model choice is still a very important factor in determining the derivative's \gd pricing bounds.  In summary, the \gd pricing bounds tell us nothing about model uncertainty itself, only about the uncertainty in the choice of a risk-neutral martingale measure within a particular model.

\citet{cochranesaa00.article} outline various ways that the \gd pricing bounds can be used, such as a trader using the bounds as buy and sell points and a bank using them as bid and ask prices for non-traded assets.  The \gd pricing bounds enable us to avoid unreasonable prices and also to examine the price sensitivity to changes in the market price of risk. 

In \citet{bayraktaryoung08.article}, Sharpe Ratios are also used to price options in incomplete markets.  However, the perspective is that of an individual seller of one option, rather than that of the entire market.  The seller of an option decides the option price via his own risk preferences, as expressed by his own chosen Sharpe Ratio.  In other words, the seller of the option chooses the risk-neutral martingale measure under which he prices the option.  It is shown in \citet{bayraktaryoung08.article} that the upper and lower \gd bounds of \citet{cochranesaa00.article} can be obtained; in that case, the seller's chosen risk-neutral martingale measure coincides with the martingale measure which gives the upper \gd bound.  The lower \gd bound is obtained in \citet{bayraktaryoung08.article} by considering the buyer of the option.

A utility-based approach to the \gd bound idea is found in \citet{cerny03.article}, and extended in \citet{kloppelschweizer07.article}.  An alternative approach based on the gain-loss ratio, which is the expectation of an asset's positive excess payoffs divided by the expectation of its negative excess payoffs, is found in \citet{bernardoledoit00.article}.

The aim of this paper is to apply the \gd bound idea to the pricing of derivatives in a regime-switching diffusion market.  The paper is structured as follows.  Section \ref{SECmarketmodel} details the regime-switching market model.  In Section \ref{SECmeasurechange} we identify the set of equivalent martingale measures.  In Section \ref{SECsharperatio} the Sharpe Ratio of an arbitrary asset in the market is defined and we state the extended \HJ bound.  The definitions of the upper and lower \gd bounds on the price of a derivative are in Section \ref{SECproblem}.  The stochastic control approach that we use to find them is detailed in Section \ref{SECstochasticcontrol}.  The minimal martingale measure, which we consider to be a benchmark pricing measure, is given in Section \ref{SECmminmartmeasure}.  A numerical example illustrating the upper and lower \gd bounds of a European call option of various maturities is in Section \ref{SECeurocall}.  We also examine the stability of the \gd bounds when we change the market model's parameters.  Finally, we conclude in Section \ref{SECconclude} with some remarks.

\section{Market model} \label{SECmarketmodel}

We consider a regime-switching diffusion market model in which there is one risk-free asset traded asset and a finite number $\BMdim$ of traded assets.  An example of a risk-free asset is a bank account and typical examples of risky assets are equities, bonds or a pooled fund.  The full mathematical description of the market is given below.

\subsection{Description of the market model}  \label{SUBSECmarketmodel}

We consider a continuous-time financial market model on a complete probability space $\probspace$ where all investment takes place over a finite time horizon $\timeset$, for a fixed $\timehorizon \in (0, \infty)$.  The probability space carries both a Markov chain $\markovchain$ and an $\BMdim$-dimensional standard Brownian motion $\vBM = (\BM_{1}, \ldots, \BM_{\BMdim})^{\top}$, where we use $A^{\top}$ to denote the transpose of $A$.

The information available to the investors in the market at time $t$ is the history of the Markov chain and Brownian motion up to and including time $t$.  Mathematically, this is represented by the filtration
\begin{equation} \label{MKTfiltration}
 \filtration_{t} := \sigma \{ (\markovchain (s), \vBM (s)), s \in [0,t] \} \vee \mathcal{N} ( \Pmeasure ), \quad \forall t \in \timeset, 
\end{equation}
where $\mathcal{N} ( \Pmeasure )$ denotes the collection of all $\Pmeasure$-null events in the probability space $\probspace$.  We assume that $\filtration = \filtration_{\timehorizon}$.

\begin{rmk} \label{RMKindepBMMC}
 As a mathematical consequence of being defined on the same filtered probability space $\filtprobspace$, the Markov chain and the Brownian motion are independent processes.  Relating these processes to economic reality, we might think of the Brownian motion as modeling short-term, micro-economic changes in the market, whereas the Markov chain models long-term macro-economic changes.  With this interpretation, the implicit assumption in our model that these economic changes are independent is a reasonable approximation to reality.  For practical implementation, this means that the number and specification of the market regimes should be chosen to reflect this interpretation.
\end{rmk}

\begin{rmk}
In our model, an investor knows what regime the market is in at each time $t$.  In reality, the market regime is unlikely to be known with certainty, although it could be estimated from market data.  This is an important point to note, since the prices of assets in the market are dependent on the initial market regime.
\end{rmk}

The market is subject to regime-switching, as modelled by the continuous-time Markov chain $\markovchain$ which takes values in a finite state space $\markovspace = \{ 1, \ldots, \markovdim \}$, for some integer $\markovdim \geq 2$.  For example, suppose we wish to model a market in which there are only two regimes: a bull market regime and a bear market regime.  We set $\markovdim=2$ and we might identify $\markovchain (t) = 1$ as corresponding to the market being in the bull market regime at time $t$.  We would then identify $\markovchain (t) = 2$ as corresponding to the market being in the bear market regime at time $t$.

We assume that the Markov chain starts in a fixed state $\markovinitialstate \in \markovspace$, so that $\markovchain(0) = \markovinitialstate$, a.s.  The Markov chain has a generator $\Markovgenerator$, which is a $\markovdim \times \markovdim$ matrix $\Markovgenerator = (\markovgenerator_{ij})_{i,j=1}^{\markovdim}$ with the properties $\markovgenerator_{ij} \geq 0$, for all $i \neq j$ and $\markovgenerator_{ii} = - \sum_{j \neq i} \markovgenerator_{ij}$.  The interpretation of the off-diagonal element $\markovgenerator_{ij}$ of the generator matrix is as the instantaneous rate of transition from state $i$ to state $j$.  To avoid states where there are no transitions into or out of, we assume that $\markovgenerator_{ii} < 0$ for each state $i$. 

Associated with each pair of distinct states $(i,j)$ in the state space of the Markov chain is a point process, or counting process, 
\begin{equation} \label{EQNmarkovpointprocess}
\markovsquareqvprocess_{ij} (t) :=
\sum_{0 < s \leq t} \indicatorfn_{ \{ \markovchain (s_{-}) = i \} } \, \indicatorfn_{ \{ \markovchain (s) = j \} }, \quad \forall t \in \timeset,
\end{equation}
where $\indicatorfn$ denotes the zero-one indicator function.  The process $\markovsquareqvprocess_{ij} (t)$ counts the number of jumps that the Markov chain $\markovchain$ has made from state $i$ to state $j$ up to time $t$.  Define the intensity process
\begin{equation} \label{EQNintensitymarkovmartingale}
 \markovcompensator_{ij} (t) := \markovgenerator_{ij} \, \indicatorfn_{ \{ \markovchain (t_{-}) = i \} }.
\end{equation}
If we compensate $\markovsquareqvprocess_{ij} (t)$ by $\int_{0}^{t} \markovcompensator_{ij} (s) \ud s$, then the resulting process
\begin{equation} \label{EQNcanonicalmarkovmartingale}
 \markovmartingale_{ij} (t) := \markovsquareqvprocess_{ij} (t) - \int_{0}^{t} \markovcompensator_{ij} (s) \ud s
\end{equation}
is a martingale (see \citet[Lemma IV.21.12]{rogerswilliamsii.book}).  We refer to the set of martingales $\{ \markovmartingale_{ij}; i,j \in \markovspace, i \neq j \}$ as \emph{the $\Pmeasure$-martingales of $\markovchain$}.  They are mutually orthogonal, purely discontinuous, square-integrable martingales which are null at the origin.

We consider a financial market that is built upon a finite number $\BMdim$ of traded assets, which we call risky assets, and a risk-free asset.  The risk-free rate of return in the market is denoted by the scalar stochastic process $r$ and the risk-free asset's price process $\bankaccount = \{ \bankaccount (t), t \in \timeset \}$ is governed by
\begin{equation} \label{MKTriskfreeprice}
 \frac{\ud \bankaccount (t)}{\bankaccount (t) } = r (t) \ud t, \quad \forall t \in \timeset, \quad \bankaccount (0) = 1.
\end{equation}
The mean rate of return of the $n$th risky asset is denoted by the scalar stochastic process $\meanret_{n}$ and the volatility process of the $n$th risky asset is denoted by the $\BMdim$-dimensional stochastic process $\vvol_{n} = ( \vol_{n1}, \ldots, \vol_{n \BMdim} )^{\top}$.  The price process $\stockprice_{n} = \{ \stockprice_{n} (t), t \in \timeset \}$ of the $n$th risky asset is then given by 
\begin{equation} \label{MKTriskyprice}
 \frac{\ud \stockprice_{n} (t)}{\stockprice_{n} (t) } = \meanret_{n} (t) \ud t + \vvolT_{n} (t) \ud \vBM (t), \quad \forall t \in \timeset,
\end{equation}
with the initial value $\stockprice_{n} (0)$ being a fixed, strictly positive constant in $\realnumbers$.
\begin{assumption} \label{ASSmktparameter}
The market parameters $r$, $\vmeanret = ( \meanret_{1}, \ldots, \meanret_{\BMdim})^{\top}$ and $\vvol = ( \vvol_{1}^{\top}, \ldots, \vvol_{\BMdim}^{\top})^{\top}$ are sufficiently regular to allow for the existence of a unique strong solution to (\ref{MKTriskfreeprice}) and (\ref{MKTriskyprice}).  Furthermore, the volatility process $\vvol$ is nonsingular.
\end{assumption}

\section{Martingale measures} \label{SECmeasurechange}

\subsection{Equivalent martingale measure}
From the fundamental theorem of arbitrage-free pricing, it is known that existence of an equivalent martingale measure (``EMM'') is equivalent to absence of arbitrage in the market.  Furthermore, the market is complete (in the sense that all claims can be replicated) if and only if the EMM is unique.  In our financial market model, while there is no arbitrage, the market is incomplete.  This means that while EMMs exist, there is no unique one.  This has immediate consequences for the valuation of contingent claims using our model, for example valuing European call options.  We can price a European call option by the usual risk-neutral pricing formula.  However, as there are infinitely many EMMs, we obtain a range of prices rather than a unique price.  The \gd bound approach is a means of narrowing the range of prices, which can be too wide to be useful in practice.  The essential idea is to exclude those EMMs which imply a Sharpe Ratio that is too high.

\subsubsection{The Girsanov kernel process and the Girsanov Theorem}

Suppose we are given a probability measure $\Qmeasure$ on $(\measurablespace)$ which is equivalent to the (real-world) probability measure $\Pmeasure$.  We define the likelihood process corresponding to the measure $\Qmeasure$ in the usual way as
\begin{displaymath}
\RNderivative (t) := \expectation \left( \frac{\ud \Qmeasure}{\ud \Pmeasure} \, \bigg\vert \, \filtration_{t} \right), \quad \forall t \in \timeset.
\end{displaymath}
We can assume that $\RNderivative$ is a positive $\{ \filtration_{t} \}$-martingale under the measure $\Pmeasure$ (see \citet[Theorem IV.17.1]{rogerswilliamsii.book}) with $\RNderivative (0) = 1$, $\Pmeasure$-a.s.  Recalling that the filtration is generated by both the Brownian motion $\vBM$ and the Markov chain $\markovchain$, we can apply an appropriate martingale representation theorem (for example, see \citet[Theorem 5.1]{elliott76.article}) to obtain predictable and suitably integrable stochastic processes $(\vRNforBM,\vRNforMC)$, for $\vRNforBM = (\RNforBM_{1}, \ldots, \RNforBM_{\BMdim})^{\top}$ and $\vRNforMC := \{ \RNforMC_{ij}; \, i,j=1,\ldots,\markovdim, j \neq i \}$, satisfying
\begin{equation} \label{EQNRNderdefn}
 \frac{\ud \RNderivative (t)}{\RNderivative(t_{-})} = \vRNforBMT (t) \ud \vBM (t) + \sum_{i=1}^{\markovdim} \sum_{\substack{j=1, \\ j \neq i}}^{\markovdim} \RNforMC_{ij} (t) \ud \markovmartingale_{ij} (t), \quad \forall t \in \timeset.
\end{equation}
In order that the measure $\Qmeasure$ is non-negative, the process $\vRNforMC$ must satisfy
\begin{displaymath}
  \RNforMC_{ij} (t) \geq -1, \qquad \forall j \neq i, \quad \forall t \in \timeset. 
\end{displaymath}
We call $(\vRNforBM,\vRNforMC)$ a \emph{Girsanov kernel process}.  As a consequence of the Girsanov theorem (for example, see \citet[Theorem 40, page 135]{protter.book}),
\begin{itemize}
 	\item we have
		\begin{equation} \label{EQNmeasurechangeBM}
			\ud \vPBM (t) =  \vRNforBMT (t) \ud t + \ud \vQBM (t),
		\end{equation}
		where, by L{\'e}vy's Theorem, $\vQBM$ is a $\Qmeasure$-Brownian motion; and
	\item the process 
		\begin{equation} \label{EQNmeasurechangemart}
			\Qmarkovmartingale_{ij} (t) := \markovsquareqvprocess_{ij} (t) - \int_{0}^{t} \left( 1 + \RNforMC_{ij} (s) \right) \markovcompensator_{ij} (s) \ud s, \quad \forall t \in \timeset,
		\end{equation}
		is a $\Qmeasure$-martingale, for each $j \neq i$.  We can interpret $\left( 1 + \RNforMC_{ij} (t) \right) \markovcompensator_{ij} (t)$ as the intensity of the point process $\markovsquareqvprocess_{ij} (t)$ under the measure $\Qmeasure$.
\end{itemize}
The set of martingales $\{ \Qmarkovmartingale_{ij}; i,j \in \markovspace, j \neq i \}$ are the $\Qmeasure$-martingales of $\markovchain$.  They are mutually orthogonal, purely discontinuous martingales which are null at the origin.  Their integrability depends on the integrability of $\RNforMC_{ij} (t)$.  Furthermore, substituting from (\ref{EQNcanonicalmarkovmartingale}) into (\ref{EQNmeasurechangemart}), we find
\begin{equation} \label{EQNmeasurechangeMC}
\ud \markovmartingale_{ij} (t) = \RNforMC_{ij} (t)\markovcompensator_{ij} (t) \ud t + \ud \Qmarkovmartingale_{ij} (t),
\end{equation}
which is analogous to (\ref{EQNmeasurechangeBM}).
\begin{rmk}
 While $\markovchain$ retains the Markov property under the measure $\Qmeasure$ (this can be shown using martingale problems, for example see \citet[Theorem 4.4.1]{ethierkurtz.book}), it is not generally a Markov chain.  This is because the intensity $\left( 1 + \RNforMC_{ij} (t) \right) \markovcompensator_{ij} (t)$ of the point process under the measure $\Qmeasure$ is not generally deterministic.
\end{rmk}

\begin{rmk} \label{RMKopenintervalgd}
The condition that $\RNforMC_{ij} (t) \geq -1$ is to ensure that the measure $\Qmeasure$ is non-negative.  However, if $\RNforMC_{ij} (t) = -1$ then $\Pmeasure$ and $\Qmeasure$ are not necessarily equivalent which means that we can have arbitrage.  As discussed in \citet[Remark 3.3]{bjorkslinko06.article}, to avoid any arbitrage possibility we can replace the inequality $\RNforMC_{ij} (t) \geq -1$ by $\RNforMC_{ij} (t) \geq -1 + \epsilon$, for some fixed $0 < \epsilon \ll 1$ or we can regard any good-deal bounds derived with the constraint $\RNforMC_{ij} (t) \geq -1$ as open intervals of good-deal bounds.  We choose the latter alternative since it is mathematically more convenient.
\end{rmk}

\subsubsection{The martingale condition and admissible Girsanov kernel processes}
Given a suitable process $(\vRNforBM,\vRNforMC)$, we can generate a corresponding measure $\Qmeasure$ by using (\ref{EQNRNderdefn}) to define the likelihood process $\RNderivative$ and then constructing the measure $\Qmeasure$ by
\begin{equation} \label{EQNQmeasuredefn}
 \frac{\ud \Qmeasure}{\ud \Pmeasure} = \RNderivative (t), \quad \textrm{on $\filtration_{t}$}.
\end{equation}
Let $\Qmeasure$ be the measure generated by the Girsanov kernel process $(\vRNforBM,\vRNforMC)$.  Consider an arbitrary asset in the market, with price process $\geninstrumentprice = \{ \geninstrumentprice (t); t \in \timeset \}$.  Note that this asset is not restricted to the traded risky assets or risk-free asset, but it could be any derivative or self-financing strategy based on them and the Markov chain $\markovchain$.  If we price this asset using a risk-neutral measure $\Qmeasure$, then the discounted price process is an $\{ \filtration_{t} \}$-martingale under the measure $\Qmeasure$.  As the filtration $\{ \filtration_{t} \}$ is generated by both the Brownian motion and the Markov chain (recall (\ref{MKTfiltration})), then using a suitable martingale representation theorem (such as \citet[Theorem 5.1]{elliott76.article}) this $\{ \filtration_{t} \}$-martingale can be expressed as the sum of a stochastic integral with respect to the Brownian motion and a stochastic integral with respect to the $\Qmeasure$-martingales of the Markov chain.  If we use the Girsanov theorem to obtain the $\Pmeasure$-dynamics of the price process, we still have a term involving the martingales of the Markov chain.  This is the reason why the $\Pmeasure$-dynamics of the price process $\geninstrumentprice$ are of the form
\begin{equation} \label{EQNgeninstrumentPdynamics}
  \frac{\ud \geninstrumentprice (t)}{\geninstrumentprice (t_{-})} = \geninstrumentmeanret (t) \ud t + \left( \vgeninstrumentvolBM (t) \right)^{\top} \ud \vPBM (t) + \sum_{i=1}^{\markovdim} \sum_{\substack{j=1, \\ j \neq i}}^{\markovdim} \geninstrumentvolMC_{ij} (t) \ud \Pmarkovmartingaleij (t).
\end{equation}
The processes $\geninstrumentmeanret$, $\vgeninstrumentvolBM = (\geninstrumentvolBM_{1}, \ldots, \geninstrumentvolBM_{\BMdim})^{\top}$ and $(\geninstrumentvolMC_{ij})_{j \neq i}$ are suitably integrable and measurable with the condition, in order to avoid negative asset prices, that $\geninstrumentvolMC_{ij} (t) \geq -1$ for each $j \neq i$.  Note that if the asset is not the traded asset then the processes $\geninstrumentmeanret$,  $\vgeninstrumentvolBM$ and $(\geninstrumentvolMC_{ij})_{j \neq i}$ depend on the choice of the risk-neutral measure through the corresponding Girsanov kernel process.

Apply (\ref{EQNmeasurechangeBM}) and (\ref{EQNmeasurechangeMC}) to (\ref{EQNgeninstrumentPdynamics}) to obtain the price dynamics $\geninstrumentprice$ of the arbitrarily chosen asset under the measure $\Qmeasure$:
\begin{equation} \label{EQNgeninstrumentQdynamics}
\begin{split}
 \frac{\ud \geninstrumentprice (t)}{\geninstrumentprice (t_{-})} & = \left( \geninstrumentmeanret (t) + \vRNforBMT (t) \vgeninstrumentvolBM (t) + \sum_{i=1}^{\markovdim} \sum_{\substack{j=1, \\ j \neq i}}^{\markovdim} \geninstrumentvolMC_{ij} (t) \RNforMC_{ij} (t) \, \markovcompensator_{ij} (t) \right) \ud t \\
& \qquad + \left( \vgeninstrumentvolBM (t) \right)^{\top} \ud \vQBM (t) + \sum_{i=1}^{\markovdim} \sum_{\substack{j=1, \\ j \neq i}}^{\markovdim} \geninstrumentvolMC_{ij} (t) \ud \Qmarkovmartingale_{ij} (t).
\end{split}
\end{equation}
The measure $\Qmeasure$ is a martingale measure if and only if the local rate of return of the asset under the measure $\Qmeasure$ equals the risk-free rate of return $r$.  Thus we obtain the following martingale condition.
\begin{propn}\emph{Martingale condition} \label{PROPNmartcondition}
The measure $\Qmeasure$ generated by the Girsanov kernel process $(\vRNforBM,\vRNforMC)$ is a martingale measure if and only if
\begin{equation}
\RNforMC_{ij} (t) \geq -1, \quad \forall j \neq i,
\end{equation}
and for any asset in the market whose price process $\geninstrumentprice$ has $\Pmeasure$-dynamics given by (\ref{EQNgeninstrumentPdynamics}), we have
\begin{equation} \label{EQNriskfreeequality}
r(t) = \geninstrumentmeanret (t) + \vRNforBMT (t) \vgeninstrumentvolBM (t)  + \sum_{i=1}^{\markovdim} \sum_{\substack{j=1, \\ j \neq i}}^{\markovdim} \geninstrumentvolMC_{ij} (t) \RNforMC_{ij} (t) \, \markovcompensator_{ij} (t), \quad \forall t \in \timeset.
\end{equation}
\end{propn}
We refer to a Girsanov kernel process $(\vRNforBM,\vRNforMC)$ for which the generated measure $\Qmeasure$ is a martingale measure as an \emph{admissible} Girsanov kernel process.

\begin{rmk}
From (\ref{EQNriskfreeequality}) we have the following economic interpretation of an admissible Girsanov kernel process $(\vRNforBM, \vRNforMC)$: the process $-\vRNforBM$ is the market price of diffusion risk and $-\RNforMC_{ij}$ is the market price of jump risk, for a jump in the Markov chain from state $i$ to state $j$.
\end{rmk}

Suppose we are given a Girsanov kernel process $(\vRNforBM,\vRNforMC)$ for which the generated measure $\Qmeasure$ is a martingale measure.  The price dynamics under $\Pmeasure$ of the $n$th underlying risky stock are as in (\ref{MKTriskyprice}), that is
\begin{displaymath}
 \frac{\ud \stockprice_{n} (t)}{\stockprice_{n} (t) } = \meanret_{n} (t) \ud t + \vvolT_{n} (t) \ud \vBM (t), \quad \forall t \in \timeset.
\end{displaymath}
By Proposition \ref{PROPNmartcondition}, we must have that
\begin{displaymath}
 r(t) = \meanret_{n} (t) + \vRNforBM^{\top} (t) \vvol_{n} (t), \quad \forall t \in \timeset, \quad \textrm{for $n=1,\ldots,\BMdim$}.
\end{displaymath}
This means that the market price of diffusion risk $-\vRNforBM$ is determined by the price dynamics of the underlying risky assets, with the solution given by
\begin{displaymath}
 \vRNforBM (t) = - \left( \vvolT (t) \right)^{-1} \left( \vmeanret (t) - r (t) \mathbf{1} \right),
\end{displaymath}
where $\mathbf{1} \in \realnumbersN$ has all entries equal to one.  However, as there is no traded asset in the market which is based on the Markov chain, we cannot say anything about the market price of jump risk $-\RNforMC_{ij}$.

\section{The Sharpe Ratio and a \HJ Bound} \label{SECsharperatio}

\subsection{The Sharpe Ratio of an arbitrary asset}
We define a Sharpe Ratio process for an arbitrarily chosen asset, with $\Pmeasure$-dynamics as in (\ref{EQNgeninstrumentPdynamics}).  Broadly, the Sharpe Ratio is the excess return above the risk-free rate of the asset per unit of risk.  We make this definition precise in our model.   Define a \emph{volatility process} $\geninstrumentvolatility$ for the asset by
\begin{equation} \label{EQNgeninstrumentvolatilityeqn}
 \ud \langle \geninstrumentprice, \geninstrumentprice \rangle (t) = \geninstrumentprice^{2} (t_{-}) \geninstrumentvolatility^{2} (t) \ud t,
\end{equation}
where $\langle \cdot , \cdot \rangle$ is the angle-bracket process.  Substituting for $\geninstrumentprice$ from (\ref{EQNgeninstrumentPdynamics}) and using $\lVert \cdot \rVert$ to denote the usual Euclidean norm, we obtain
\begin{equation} \label{EQNgeninstrumentvolatilitytwin}
\ud \langle \geninstrumentprice, \geninstrumentprice \rangle (t) = \geninstrumentprice^{2} (t_{-}) \left( \lVert \vgeninstrumentvolBM (t) \rVert^{2} + \sum_{i=1}^{\markovdim} \sum_{\substack{j=1, \\ j \neq i}}^{\markovdim} \big\lvert \geninstrumentvolMC_{ij} (t) \big\rvert^{2} \, \markovcompensator_{ij} (t) \right) \ud t.
\end{equation}
Comparing (\ref{EQNgeninstrumentvolatilityeqn}) and (\ref{EQNgeninstrumentvolatilitytwin}), we see that the squared volatility process satisfies 
\begin{displaymath}
 \geninstrumentvolatility^{2} (t) = \lVert \vgeninstrumentvolBM (t) \rVert^{2} +  \sum_{i=1}^{\markovdim} \sum_{\substack{j=1, \\ j \neq i}}^{\markovdim} \lvert \geninstrumentvolMC_{ij} (t) \rvert^{2} \, \markovcompensator_{ij} (t).
\end{displaymath}
Recalling that the state space of the Markov chain $\markovchain$ is denoted by $\markovspace = \{1, \ldots, \markovdim \}$ and the intensity process $\markovcompensator_{ij} (t)$ is given by (\ref{EQNintensitymarkovmartingale}), define the norm $\lVert \cdot \rVert_{\markovcompensator (t)}$ in the Hilbert space $L^{2} (\markovspace \times \markovspace, \markovcompensator (t) )$ by
\begin{displaymath}
 \lVert \pmb{\gamma} (t) \rVert^{2}_{\markovcompensator (t)} := \sum_{i=1}^{\markovdim} \sum_{\substack{j=1, \\ j \neq i}}^{\markovdim} \lvert \gamma_{ij} (t) \rvert^{2} \, \markovcompensator_{ij} (t).
\end{displaymath}
Then we can write
\begin{displaymath}
 \geninstrumentvolatility^{2} (t) = \lVert \vgeninstrumentvolBM (t) \rVert^{2} + \lVert \vgeninstrumentvolMC (t) \rVert^{2}_{\markovcompensator (t)}.
\end{displaymath}
Defining the Hilbert space
\begin{equation} \label{EQNhilbertspace}
\hilbertspace := \realnumbersN \times L^{2} (\markovspace \times \markovspace, \markovcompensator (t) ),
\end{equation}
and denoting by $\lVert \cdot \rVert_{\hilbertspace}$ the norm in the Hilbert space $\hilbertspace$, we can also express the volatility process as
\begin{equation} \label{EQNvolatilityprocess}
 \geninstrumentvolatility (t) = \lVert \left( \vgeninstrumentvolBM (t), \vgeninstrumentvolMC (t) \right) \rVert_{\hilbertspace}.
\end{equation}

Finally, we are in a position to define the \emph{Sharpe Ratio process} $\sharperatio$ for the arbitrarily-chosen asset.  As $\geninstrumentmeanret$ is the local mean rate of return of the asset under the measure $\Pmeasure$,
\begin{equation} \label{EQNsharperatio}
 \sharperatio (t) := \frac{\geninstrumentmeanret (t) - r (t)}{\geninstrumentvolatility (t)}.
\end{equation}
The Sharpe Ratio process depends on the chosen asset's price process.  However, we seek a bound that applies to all assets' Sharpe Ratio processes.  To do this, we use the extended \HJ inequality, which is derived in \citet{bjorkslinko06.article} and is an extended version of the inequality introduced by \citet{hansenjagannathan91.article}.

\subsection{An extended \HJ Bound}

\begin{lma}[An extended \HJ Bound] \label{LMAHJbound}
Recall the Hilbert space $\hilbertspace$ in (\ref{EQNhilbertspace}).  For every admissible Girsanov kernel process $(\vRNforBM, \vRNforMC)$ and for any asset in the market whose price process $\geninstrumentprice$ has $\Pmeasure$-dynamics given by (\ref{EQNgeninstrumentPdynamics}) and, consequently, whose Sharpe Ratio process $\sharperatio$ is given by (\ref{EQNsharperatio}), the following inequality holds.
\begin{displaymath}
 \lvert \sharperatio (t) \rvert \leq \lVert \left( \vRNforBM (t), \vRNforMC (t) \right) \rVert_{\hilbertspace},
\end{displaymath}
that is
\begin{equation} \label{EQNSRextHKbnd}
 \lvert \sharperatio (t) \rvert^{2} \leq \lVert \vRNforBM (t) \rVert^{2} + \sum_{i=1}^{\markovdim} \sum_{\substack{j=1, \\ j \neq i}}^{\markovdim} \lvert \RNforMC_{ij} (t) \rvert^{2} \, \markovcompensator_{ij} (t).
\end{equation}
\end{lma}
\begin{proof}
The proof follows that of \citet[Theorem A.1]{bjorkslinko06.article} and is therefore omitted.
\end{proof}
From Lemma \ref{LMAHJbound}, we see that we can bound the Sharpe Ratios of all assets in the market by bounding the right-hand side of (\ref{EQNSRextHKbnd}) by a constant.

\section{The general problem} \label{SECproblem}
On the market detailed in Subsection \ref{SUBSECmarketmodel}, we consider the valuation of a general contingent claim.  To apply the \gd bound idea, suppose we are given a contingent claim $\contingentclaim$ of the form
\begin{equation}
 \contingentclaim := \contingentfn (\vstockprice (\timehorizon), \markovchain (\timehorizon)),
\end{equation}
for a deterministic, measurable function $\Phi$, where $\vstockprice = (\stockprice_{1}, \ldots, \stockprice_{\BMdim})^{\top}$ is the vector of the risky assets' price processes.  As there is no unique martingale measure in the market, there is no unique price for the contingent claim.  Rather than choosing one particular martingale measure to price the contingent claim, we seek instead to find a reasonable range of prices by excluding those martingale measures which imply Sharpe Ratios which are too high.

\subsection{The \gd bound}
The key idea is that to restrict the set of martingale measures by way of the Sharpe Ratio, we use the \HJ bound.  Rather than bounding the Sharpe Ratios directly, we bound the right-hand side of (\ref{EQNSRextHKbnd}) by a constant.  We call the constant a \emph{\gd bound}.

\begin{cond} \label{CONDmingdb}
There exists $\gooddealbnd_{0} \in \realnumbers$ such that
\begin{displaymath}
 \gooddealbnd_{0} = \sup_{t \in \timeset} \lVert \vRNforBM (t) \rVert^{2}, \qquad \textrm{a.s.}
\end{displaymath}
\end{cond}
\begin{defn} \label{DEFNgdb}
 A \gd bound is a constant $\gooddealbnd \geq \gooddealbnd_{0}$.
\end{defn}
\begin{rmk} \label{RMKgdinterpret}
 A chosen \gd bound $\gooddealbnd$ bounds the Sharpe Ratio process $\sharperatio$ of any asset in the market as follows:
\begin{equation} \label{EQNgdbinequality}
 \lvert \sharperatio (t) \rvert^{2} \leq \lVert \vRNforBM (t) \rVert^{2} + \sum_{i=1}^{\markovdim} \sum_{\substack{j=1, \\ j \neq i}}^{\markovdim} \lvert \RNforMC_{ij} (t) \rvert^{2} \, \markovcompensator_{ij} (t) \leq \gooddealbnd.
\end{equation}
In other words, $\lvert \sharperatio (t) \rvert \leq \sqrt{\gooddealbnd}$.  The economic interpretation is that, under the \gd bound approach, $\sqrt{\gooddealbnd}$ and $- \sqrt{\gooddealbnd}$ are the highest and lowest achievable instantaneous Sharpe Ratio in the market, respectively.  However, in the regime-switching diffusion market, we see from (\ref{EQNgdbinequality}) that the \gd bound $\gooddealbnd$ is really a bound on the price $-\RNforMC_{ij}$ of regime change risk, since the price $-\vRNforBM$ of diffusion risk is determined by the traded assets.
\end{rmk}

\subsection{The \gd bound price processes}
We consider the problem of finding the upper and lower \gd bounds on the range of possible prices of the contingent claim $\contingentclaim$.
\begin{defn} \label{DEFNuppergooddeal}
Suppose we are given a \gd bound $\gooddealbnd$.  The \emph{upper \gd price process} $\valuefnupper$ for the bound $\gooddealbnd$ is the optimal value process for the control problem
\begin{equation} \label{EQNclaimprice}
 \sup_{(\vRNforBM, \vRNforMC)} \expectationQ \left( e^{ - \int_{t}^{\timehorizon} r (\tau) \ud \tau } \contingentfn (\vstockprice (\timehorizon), \markovchain (\timehorizon)) \, \bigg\vert \, \filtration_{t} \right),
\end{equation}
where the predictable processes $(\vRNforBM, \vRNforMC)$ are subject to the constraints
\begin{equation} \label{EQNconstraintBM}
 \vRNforBM (t) = - \left( \vvolT (t) \right)^{-1} \left( \vmeanret (t) - r (t) \mathbf{1} \right),
\end{equation}
\begin{equation} \label{EQNconstraintMC}
\RNforMC_{ij} (t) \geq -1, \qquad \textrm{for $i,j=1, \ldots, \markovdim$, $j \neq i$},
\end{equation}
and
\begin{equation} \label{EQNconstraintGDB}
 \lVert \vRNforBM (t) \rVert^{2} + \sum_{i=1}^{\markovdim} \sum_{\substack{j=1, \\ j \neq i}}^{\markovdim} \lvert \RNforMC_{ij} (t) \rvert^{2} \, \markovcompensator_{ij} (t) \leq \gooddealbnd,
\end{equation}
for all $t \in \timeset$.
\end{defn}
\begin{defn} \label{DEFNlowergooddeal}
 The \emph{lower \gd price process} $\valuefnlower$ is defined as in Definition \ref{DEFNuppergooddeal} except that ``sup'' in (\ref{EQNclaimprice}) is replaced by ``inf''.
\end{defn}
\begin{rmk}
The risk-neutral valuation formula in (\ref{EQNclaimprice}) implies that the local rate of return of the price process corresponding to the contingent claim $\contingentclaim = \contingentfn (\vstockprice (\timehorizon), \markovchain (\timehorizon))$ equals the risk-free rate $r$ under the measure $\Qmeasure$.  The equality constraint (\ref{EQNconstraintBM}) ensures that $\vRNforBM$ is consistent with the market price of jump risk.  Together with the constraint (\ref{EQNconstraintMC}), these ensure that the measure $\Qmeasure$ generated by $(\vRNforBM, \vRNforMC)$ is a martingale measure, as in Proposition \ref{PROPNmartcondition}.  Note that, due to the constant bound on $(\vRNforBM, \vRNforMC)$ in the constraint (\ref{EQNconstraintGDB}), the measure $\Qmeasure$ generated by $(\vRNforBM, \vRNforMC)$ is a martingale measure, and not just a local martingale measure.
\end{rmk}
\begin{rmk}
The constraint (\ref{EQNconstraintGDB}) arises from the \gd bound.  It is obtained by bounding the right-hand side of (\ref{EQNSRextHKbnd}).
\end{rmk}
\begin{rmk}
 The only unknown in the constraints (\ref{EQNconstraintBM})-(\ref{EQNconstraintGDB}) is the market price of jump risk $\RNforMC_{ij} (t)$.  If we obtain wide \gd pricing bounds for a derivative then this tells us that the choice of the market price of jump risk $\RNforMC_{ij} (t)$ has a large impact on the derivative's price.  Thus wide pricing bounds are a signal that we should explore additional ways of further restricting the possible values of the market price of jump risk $\RNforMC_{ij} (t)$.  This point is also made in \citet{cochranesaa00.article}.
\end{rmk}

The goal is to calculate the upper and lower \gd bound price processes, which are what we consider to be reasonable bounds on the possible prices of the contingent claim $\contingentclaim$.  To calculate them, we use a stochastic control approach.

\section{Stochastic control approach} \label{SECstochasticcontrol}
To ensure that the Markovian structure is preserved under the martingale measure $\Qmeasure$, we need the following condition.
\begin{cond} \label{CONDformofkernel}
The maximum in (\ref{EQNclaimprice}) is taken over Girsanov kernel processes $(\vRNforBM,\vRNforMC)$ of the form
\begin{displaymath}
 \vRNforBM (t) = \vRNforBM (t, \vstockprice (t), \markovchain (t_{-})) \quad \textrm{and} \quad \RNforMC_{ij} (t) = \RNforMC_{ij} (t, \vstockprice (t), \markovchain (t_{-})), \quad \forall j \neq i,
\end{displaymath}
and $\RNforMC_{ii} (t) = 0$, for all $t \in \timeset$.
\end{cond}
\begin{rmk}
We note from the constraint (\ref{EQNconstraintBM}) that the process $\vRNforBM$ is completely determined by the market parameters $r(t)$, $\vmeanret(t)$ and $\vvol (t)$.  This means that the requirement $\vRNforBM (t) = \vRNforBM (t, \vstockprice (t), \markovchain (t_{-}))$ is really a requirement that the market parameters are of the form
\begin{displaymath}
 r (t) = r (t, \vstockprice (t), \markovchain (t_{-})), \quad \vmeanret (t) = \vmeanret (t, \vstockprice (t), \markovchain (t_{-})) \quad \textrm{and} \quad \vvol (t) = \vvol (t, \vstockprice (t), \markovchain (t_{-})).
\end{displaymath}
\end{rmk}

\subsection{The \gd functions}

Under Condition \ref{CONDformofkernel}, the optimal expected value in (\ref{EQNclaimprice}) can be written as $\valuefnupper (t,\vstockprice (t), \markovchain (t_{-}))$ where the deterministic mapping $\valuefnupper : \timeset \times \realnumbersN_{+} \times \markovspace \rightarrow \realnumbers_{+}$ is known as the \emph{optimal value function}.  From general dynamic programming theory (for example, see \citet[Chapter 19]{bjork.book}), the optimal value function satisfies the following Hamilton-Jacobi-Bellman equation
\begin{align} \label{EQNHJB}
\frac{\partial \valuefn}{\partial t} + \sup_{(\vRNforBM, \vRNforMC)} \left\{ \infinitoperatorBMMC \valuefn \right\} - r \valuefn & = 0 \\
\valuefn (\timehorizon, \mathbf{x}, i) & = \contingentfn (\mathbf{x}, i), \notag
\end{align}
where the supremum in (\ref{EQNHJB}) is subject to the constraints (\ref{EQNconstraintBM}) - (\ref{EQNconstraintGDB}).  An application of It{\^o}'s formula (for example, see \citet[Theorem V.18, page 278]{protter.book}) shows that the infinitesimal operator $\infinitoperatorBMMC$ is given by
\begin{equation} \label{EQNinfinitoperator}
\begin{split}
& \infinitoperatorMC \valuefn (t, \mathbf{x}, i) \\
& = r(t,\mathbf{x},i) \sum_{n=1}^{\BMdim} x_{n} \frac{\partial \valuefn}{\partial x_{n}} (t, \mathbf{x}, i) + \frac{1}{2} \sum_{n,m=1}^{\BMdim} \vvolT_{n} (t,\mathbf{x},i)  \vvol_{m} (t,\mathbf{x},i) x_{n} x_{m} \frac{\partial^{2} \valuefn}{\partial x_{n} x_{m}} (t, \mathbf{x}, i) \\
& + \sum_{\substack{j=1, \\ j \neq i}}^{\markovdim} \Pmarkovgenerator_{ij} (1 + \RNforMC_{ij} (t, \mathbf{x}) ) \left( \valuefn (t, \mathbf{x}, j) - \valuefn (t, \mathbf{x}, i) \right),
\end{split}
\end{equation}
for all $(t,\mathbf{x},i) \in \timeset \times \realnumbersN \times \markovspace$.

\begin{defn} \label{DEFNPIDEtosolveupper}
Given a \gd bound $\gooddealbnd$, the \emph{upper \gd function} for the bound $\gooddealbnd$ is the solution to the following boundary value problem
\begin{align} \label{EQNBVP}
\frac{\partial \valuefn}{\partial t} (t, \mathbf{x}, i) + \sup_{(\vRNforBM, \vRNforMC)} \left\{ \infinitoperatorMC \valuefn (t, \mathbf{x}, i) \right\} - r (t,\mathbf{x},i) \valuefn (t, \mathbf{x}, i) & = 0 \\
\valuefn (\timehorizon, \mathbf{x}, i) & = \contingentfn (\mathbf{x}, i), \notag
\end{align}
where $\infinitoperatorMC$ is given by (\ref{EQNinfinitoperator}) and the supremum is taken over all functions $(\vRNforBM, \vRNforMC)$ subject to Condition \ref{CONDformofkernel} and satisfying
\begin{equation} \label{EQNHJBBMstd}
 \vRNforBM (t, \mathbf{x}, i) = - \left( \vvolT (t, \mathbf{x}, i) \right)^{-1} \left( \vmeanret (t, \mathbf{x}, i) - r (t, \mathbf{x}, i) \mathbf{1} \right),
\end{equation}
\begin{equation} \label{EQNHJBMCstd}
\RNforMC_{ij} (t,\mathbf{x}) \geq -1, \qquad \textrm{for $j=1, \ldots, \markovdim$, $j \neq i$},
\end{equation}
and
\begin{equation} \label{EQNHJBMCbound}
\lVert \vRNforBM (t, \mathbf{x}, i) \rVert^{2} + \sum_{\substack{j=1, \\ j \neq i}}^{\markovdim} \markovgenerator_{ij} \lvert \RNforMC_{ij} (t,\mathbf{x}) \rvert^{2} \leq \gooddealbnd,
\end{equation}
for all $(t, \mathbf{x}, i) \in \timeset \times \realnumbersN \times \markovspace$.  We denote the solution to (\ref{EQNBVP}) by $\valuefnupperii$.
\end{defn}

\begin{defn} \label{DEFNPIDEtosolvelower}
The \emph{lower \gd function} is the solution to (\ref{EQNBVP}) but with the supremum replaced by an infimum, subject to Condition \ref{CONDformofkernel} and the constraints (\ref{EQNHJBBMstd}) - (\ref{EQNHJBMCbound}).    We denote this solution by $\valuefnlowerii$.
\end{defn}
Rather than attempting to solve the partial integro-differential equation (``PIDE'') of (\ref{EQNBVP}) directly, we reduce it to two deterministic problems which we solve for each fixed triple $(t, \mathbf{x}, i) \in \timeset \times \realnumbersN \times \markovspace$.  Moreover, as $\vRNforBM$ is completely determined by (\ref{EQNHJBBMstd}), we need to solve only for the optimal $\vRNforMC$.

Therefore, given $\vRNforBM$ satisfying (\ref{EQNHJBBMstd}), we do the following.
\begin{enumerate}
 \item Solve the static optimization problem of finding the optimal $\vRNforMCopt$ in
\begin{displaymath}
 \sup_{(\vRNforBM, \vRNforMC)} \left\{ \infinitoperatorMC \valuefn (t, \mathbf{x}, i) \right\},
\end{displaymath}
subject to the constraints (\ref{EQNHJBMCstd}) and (\ref{EQNHJBMCbound}).
\item Using the optimal $\vRNforMCopt$ found above, solve the PIDE
\begin{align} \label{EQNpidetop}
\frac{\partial \valuefn}{\partial t} + \infinitoperatorMCopt \valuefn  - r \valuefn & = 0 \\ 
\valuefn (\timehorizon, \mathbf{x}, i) & = \contingentfn (\mathbf{x}, i).  \label{EQNpidebot}
\end{align}
\end{enumerate}
We consider in more detail how to solve the static optimization problem.  To solve the PIDE, we can use numerical methods.  A concrete example of this, where we find the \gd bounds for a European call option, is given in Section \ref{SECeurocall}.

\subsection{The static optimization problem}
As we have seen above, the static optimization problem associated with the upper \gd function of Definition \ref{DEFNPIDEtosolveupper} is to find for each triple $(t, \mathbf{x}, i) \in \timeset \times \realnumbersN \times \markovspace$ the optimal $\vRNforMCopt$ that attains the supremum of
\begin{equation} \label{EQNsupinfinitgen}
\begin{split}
& \infinitoperatorMC \valuefn (t, \mathbf{x}, i) \\
& = r(t,\mathbf{x},i) \sum_{n=1}^{\BMdim} x_{n} \frac{\partial \valuefn}{\partial x_{n}} (t, \mathbf{x}, i) + \frac{1}{2} \sum_{n,m=1}^{\BMdim} \vvolT_{n} (t,\mathbf{x},i)  \vvol_{m} (t,\mathbf{x},i) x_{n} x_{m} \frac{\partial^{2} \valuefn}{\partial x_{n} x_{m}} (t, \mathbf{x}, i) \\
& + \sum_{\substack{j=1, \\ j \neq i}}^{\markovdim} \Pmarkovgenerator_{ij} (1 + \RNforMC_{ij} (t, \mathbf{x}) ) \left( \valuefn (t, \mathbf{x}, j) - \valuefn (t, \mathbf{x}, i) \right),
\end{split}
\end{equation}
subject to the constraints
\begin{equation} \label{EQNHJBMCstdstatic}
\RNforMC_{ij} (t,\mathbf{x}) \geq -1, \quad \forall j \neq i \quad \textrm{and} \quad \sum_{\substack{j=1, \\ j \neq i}}^{\markovdim} \Pmarkovgenerator_{ij} \lvert \RNforMC_{ij} (t,\mathbf{x}) \rvert^{2} \leq \gooddealbnd - \lVert \vRNforBM (t, \mathbf{x}, i) \rVert^{2},
\end{equation}
with $\vRNforBM (t, \mathbf{x}, i)$ given by (\ref{EQNHJBBMstd}).

The static optimization problem associated with the lower \gd function of Definition \ref{DEFNPIDEtosolvelower} is as for the upper \gd function but taking the infimum of (\ref{EQNsupinfinitgen}), rather than the supremum.

As the only term in (\ref{EQNsupinfinitgen}) which involves $\vRNforMCopt$ is the last one, we can equivalently consider the problem of finding the optimal $\vRNforMCopt$ which attains the supremum of
\begin{equation} \label{EQNreducedstatic}
\sum_{\substack{j=1, \\ j \neq i}}^{\markovdim} \Pmarkovgenerator_{ij} (1 + \RNforMC_{ij} (t, \mathbf{x}) ) \left( \valuefn (t, \mathbf{x}, j) - \valuefn (t, \mathbf{x}, i) \right),
\end{equation}
subject to the constraints in (\ref{EQNHJBMCstdstatic}).  This is a linear optimization problem with both linear and quadratic constraints.  We consider how the complexity of this problem increases as the number of states $\markovdim$ of the Markov chain increases.

\subsubsection{Markov chain with two states}
When there are only two states of the Markov chain, the solution of the linear optimization problem (\ref{EQNreducedstatic}) subject to the constraints (\ref{EQNHJBMCstdstatic}) is very simple indeed and can be obtained by considering the sign of $\valuefn (t, \mathbf{x}, j) - \valuefn (t, \mathbf{x}, i)$ in (\ref{EQNreducedstatic}).
\begin{lma} \label{LMAstaticoptsoln}
For a $2$-state Markov chain, fix $(t, \mathbf{x}, i) \in \timeset \times \realnumbersN \times \{1,2\}$ and define
\begin{displaymath}
  \genconstraint (t, \mathbf{x}, i)  := \left( \frac{ \gooddealbnd - \lVert  \vRNforBM (t, \mathbf{x}, i) \rVert^{2} }{ - \Pmarkovgenerator_{ii} } \right)^{1/2}.
\end{displaymath}
Then for each $j \neq i$, the solution to the static optimization problem associated with the upper \gd function of Definition \ref{DEFNPIDEtosolveupper} is
\begin{displaymath}
\RNforMCopt_{ij}^{\textrm{upp}} (t, \mathbf{x}) = 
\left\{ \begin{array}{ll}
\genconstraint (t, \mathbf{x}, i) & \textrm{if $\valuefn (t, \mathbf{x}, j) - \valuefn (t, \mathbf{x}, i) > 0$} \\
- \min \left[ 1, \genconstraint (t, \mathbf{x}, i) \right] & \textrm{if $\valuefn (t, \mathbf{x}, j) - \valuefn (t, \mathbf{x}, i) \leq 0$},
\end{array}
\right.
\end{displaymath}
and the solution to the static optimization problem associated with the lower \gd function of Definition \ref{DEFNPIDEtosolvelower} is
\begin{displaymath}
\RNforMCopt_{ij}^{\textrm{low}} (t, \mathbf{x}) = 
\left\{ \begin{array}{ll}
 - \min \left[ 1, \genconstraint (t, \mathbf{x}, i) \right] & \textrm{if $\valuefn (t, \mathbf{x}, j) - \valuefn (t, \mathbf{x}, i) > 0$} \\
\genconstraint (t, \mathbf{x}, i) & \textrm{if $\valuefn (t, \mathbf{x}, j) - \valuefn (t, \mathbf{x}, i) \leq 0$}.
\end{array}
\right.
\end{displaymath}
\end{lma}
Thus the solutions $\RNforMCopt_{ij}^{\textrm{low}} (t, \mathbf{x})$ and $\RNforMCopt_{ij}^{\textrm{upp}} (t, \mathbf{x})$ depend on the value function $\valuefn$.   This means that the numerical solution to the PIDE (\ref{EQNpidetop})-(\ref{EQNpidebot}) involves checking at each node of the discretized state space the relative values of the value function in order to choose the appropriate solution. 

\subsubsection{Markov chain with three or more states}
For a Markov chain with more than two states, the solution becomes more complicated because the number of constraints increases.  If there are $\markovdim$ states then for each fixed triple $(t, \mathbf{x}, i) \in \timeset \times \realnumbersN \times \markovspace$ there are $\markovdim-1$ variables $\{ \RNforMC_{ij} (t, \mathbf{x}) : j=1, \ldots, \markovdim, j \neq i \}$ to find, each of which is subject to a lower and upper inequality constraint.  Hence there are $2^{\markovdim-1}$ potential solutions, depending on which of the lower and upper constraints is binding.  To see how the complexity increases, we consider a Markov chain with three states.  For each $x \in \realnumbers$, denote by $\sgn(x)$ the sign of $x$.  
\begin{lma} \label{LMAstaticoptsolnthree}
For a $3$-state Markov chain, fix $(t, \mathbf{x}, i) \in \timeset \times \realnumbersN \times \{1,2,3 \}$ and define for each $j \in \{1,2,3 \}$, $j \neq i$,
\begin{displaymath}
  \genconstraint (t, \mathbf{x}, i, j)  := \left( \frac{ \gooddealbnd - \lVert  \vRNforBM (t, \mathbf{x}, i) \rVert^{2} + \Pmarkovgenerator_{ii} + \Pmarkovgenerator_{ij} }{ - \Pmarkovgenerator_{ij} } \right)^{1/2}.
\end{displaymath}
Then for each $j \neq i$ and $k \neq i$, $k \neq j$, the solution $( \RNforMCopt_{ij}^{\textrm{upp}} (t, \mathbf{x}), \RNforMCopt_{ik}^{\textrm{upp}} (t, \mathbf{x}) )$ to the static optimization problem associated with the upper \gd function of Definition \ref{DEFNPIDEtosolveupper} is one of the following pairs:
\begin{equation} \label{EQNuppthreeMCone}
(-1, \, -1),
\end{equation}
\begin{equation} \label{EQNuppthreeMCtwo}
(\sgn \left( \valuefn (t, \mathbf{x}, j) - \valuefn (t, \mathbf{x}, i) \right) \genconstraint (t, \mathbf{x}, i, j), \, -1),
\end{equation}
\begin{equation} \label{EQNuppthreeMCthree}
(-1, \, \sgn \left( \valuefn (t, \mathbf{x}, k) - \valuefn (t, \mathbf{x}, i) \right)\genconstraint (t, \mathbf{x}, i, k)),
\end{equation}
\begin{equation} \label{EQNuppthreeMCfour}
\begin{split}
\Bigg( & \left( \valuefn (t, \mathbf{x}, j) - \valuefn (t, \mathbf{x}, i) \right) \left( \frac{ \gooddealbnd - \lVert  \vRNforBM (t, \mathbf{x}, i) \rVert^{2} }{ \sum_{\ell = 1}^{\markovdim} \Pmarkovgenerator_{i \ell} \left(  \valuefn (t, \mathbf{x}, \ell) - \valuefn (t, \mathbf{x}, i) \right)^{2} } \right)^{1/2}, \\
& \left( \valuefn (t, \mathbf{x}, k) - \valuefn (t, \mathbf{x}, i) \right) \left( \frac{ \gooddealbnd - \lVert  \vRNforBM (t, \mathbf{x}, i) \rVert^{2} }{  \sum_{\ell = 1}^{\markovdim} \Pmarkovgenerator_{i \ell} \left(  \valuefn (t, \mathbf{x}, \ell) - \valuefn (t, \mathbf{x}, i) \right)^{2} } \right)^{1/2} \Bigg).
\end{split}
\end{equation}
and the solution $( \RNforMCopt_{ij}^{\textrm{low}} (t, \mathbf{x}), \RNforMCopt_{ik}^{\textrm{low}} (t, \mathbf{x}) )$ to the static optimization problem associated with the lower \gd function of Definition \ref{DEFNPIDEtosolvelower} is one of the following pairs:
\begin{equation} \label{EQNlowthreeMCone}
(-1, \, -1),
\end{equation}
\begin{equation} \label{EQNlowthreeMCtwo}
(- \sgn \left( \valuefn (t, \mathbf{x}, j) - \valuefn (t, \mathbf{x}, i) \right) \genconstraint (t, \mathbf{x}, i, j), \, -1),
\end{equation}
\begin{equation} \label{EQNlowthreeMCthree}
(-1, \, - \sgn \left( \valuefn (t, \mathbf{x}, k) - \valuefn (t, \mathbf{x}, i) \right)\genconstraint (t, \mathbf{x}, i, k)),
\end{equation}
\begin{equation} \label{EQNlowthreeMCfour}
\begin{split}
\Bigg( & - \left( \valuefn (t, \mathbf{x}, j) - \valuefn (t, \mathbf{x}, i) \right) \left( \frac{ \gooddealbnd - \lVert  \vRNforBM (t, \mathbf{x}, i) \rVert^{2} }{ \sum_{\ell = 1}^{\markovdim} \Pmarkovgenerator_{i \ell} \left(  \valuefn (t, \mathbf{x}, \ell) - \valuefn (t, \mathbf{x}, i) \right)^{2} } \right)^{1/2}, \\
& - \left( \valuefn (t, \mathbf{x}, k) - \valuefn (t, \mathbf{x}, i) \right) \left( \frac{ \gooddealbnd - \lVert  \vRNforBM (t, \mathbf{x}, i) \rVert^{2} }{  \sum_{\ell = 1}^{\markovdim} \Pmarkovgenerator_{i \ell} \left(  \valuefn (t, \mathbf{x}, \ell) - \valuefn (t, \mathbf{x}, i) \right)^{2} } \right)^{1/2} \Bigg).
\end{split}
\end{equation}
\end{lma}
\begin{proof}
 Apply the Kuhn-Tucker method.
\end{proof}
\begin{rmk}
The solutions $( \RNforMCopt_{ij}^{\textrm{upp}} (t, \mathbf{x}), \RNforMCopt_{ik}^{\textrm{upp}} (t, \mathbf{x}) )$ and $( \RNforMCopt_{ij}^{\textrm{low}} (t, \mathbf{x}), \RNforMCopt_{ik}^{\textrm{low}} (t, \mathbf{x}) )$ depend on the value function $\valuefn$,  just as in the two-state Markov chain case.  However, the difficulty involved in solving the PIDE (\ref{EQNpidetop})-(\ref{EQNpidebot}) numerically has increased since there are four potential solutions which must be checked at each node of the discretized state space.  As we increase the number of states in the Markov chain, the number of potential solutions to the static optimization problem increases and hence the complexity involved in solving the PIDE increases too.
\end{rmk}

\section{Minimal martingale measure} \label{SECmminmartmeasure}
Here we leave aside the \gd bounds and consider the minimal martingale measure, which we consider as a benchmark measure for pricing any derivative in the market.
\begin{defn} \label{PROBmmm}
 The \emph{minimal martingale measure} is the measure $\Qmeasuremmm$ generated by $(\vRNforBMmmm, \vRNforMCmmm)$, where $(\vRNforBMmmm, \vRNforMCmmm)$ is the Girsanov kernel process which minimizes
 \begin{displaymath}
\lVert \vRNforBM (t) \rVert^{2} + \sum_{i=1}^{\markovdim} \sum_{\substack{j=1, \\ j \neq i}}^{\markovdim} \lvert \RNforMC_{ij} (t) \rvert^{2} \, \markovcompensator_{ij} (t)
\end{displaymath}
subject to the constraint $r(t,\mathbf{x},i) = \meanret_{n} (t,\mathbf{x},i) + \vRNforBM^{\top} (t,\mathbf{x},i) \vvol_{n} (t,\mathbf{x},i)$ for $n=1,\ldots,\BMdim$.
\end{defn}

It is immediate that the minimal martingale measure $\Qmeasuremmm$ is generated by
\begin{displaymath}
 \vRNforBMmmm (t) := - \left( \vvolT (t) \right)^{-1} \left( \vmeanret(t) - r(t) \mathbf{1} \right) \qquad \textrm{and} \qquad \RNforMCmmm_{ij} (t) := 0, \quad \forall j \neq i,
\end{displaymath}
for all $t \in \timeset$, where $\mathbf{1} \in \realnumbersN$ has all entries equal to one.  As $\RNforMCmmm_{ij} (t) \geq -1$, we have that $(\vRNforBMmmm, \vRNforMCmmm)$ is an admissible Girsanov kernel process.

\begin{rmk}
Under the measure $\Qmeasuremmm$, the process $\markovchain$ is a Markov chain with the same generator $\Markovgenerator = (\markovgenerator_{ij})_{i,j=1}^{\markovdim}$ as under the measure $\Pmeasure$.  In particular, this means that the measure $\Qmeasuremmm$ preserves the martingale property of the process $\markovmartingaleij (t)$ defined by (\ref{EQNcanonicalmarkovmartingale}), so that the $\Pmeasure$-martingales of $\markovchain$ are also its $\Qmeasuremmm$-martingales.
\end{rmk}

Notice that $(\vRNforBMmmm, \vRNforMCmmm)$ minimizes the right-hand side of (\ref{EQNSRextHKbnd}) over the set of admissible Girsanov kernel processes.  Moreover, by Definition \ref{DEFNgdb}, any \gd bound $\gooddealbnd$ satisfies $\gooddealbnd \geq \gooddealbnd_{0}$.  This means
\begin{displaymath}
\gooddealbnd \geq \gooddealbnd_{0} = \sup_{t \in \timeset} \lVert \vRNforBM (t) \rVert^{2} = \sup_{t \in \timeset} \lVert \left( \vRNforBMmmm (t), \vRNforMCmmm (t) \right) \rVert^{2}.
\end{displaymath}
Thus $(\vRNforBMmmm, \vRNforMCmmm)$ is a Girsanov kernel process which satisfies the \gd bound constraint in (\ref{EQNHJBMCbound}).

Denote the solution to the PIDE
\begin{align}
\frac{\partial \valuefn}{\partial t} + \infinitoperatormmm \valuefn - r \valuefn & = 0 \\
\valuefn (\timehorizon, \mathbf{x}, i) & = \contingentfn (\mathbf{x}, i)
\end{align}
by $\valuefnmmm$.  Then as $(\vRNforBMmmm, \vRNforMCmmm)$ is a Girsanov kernel process which satisfies the \gd bound constraint in (\ref{EQNHJBMCbound}), it is clear from this and Definitions \ref{DEFNPIDEtosolveupper} and \ref{DEFNPIDEtosolvelower} that the following relation holds:
\begin{displaymath}
\valuefnlowerii \leq \valuefnmmm \leq \valuefnupperii.
\end{displaymath}

\section{Numerical example} \label{SECeurocall}

Having applied the \gd bound idea in a regime-switching diffusion market, the next question is: are they useful?  We examine this question by calculating the upper and lower \gd pricing bounds for a $1$-year European call option in a market where there are two regimes.  We calculate them for various initial stock prices and for various choices of the \gd bound.  Finally, we examine how the pricing bounds change as we change the generator of the Markov chain which drives the regime-switching.
  
\subsection{Market model} \label{SUBSECmodelone}
Suppose that we are in a financial market setting of Section \ref{SECmarketmodel} with only two market regimes and one risky asset, so that $\BMdim = 1$, and time is measured in years.  Assume the values of the market parameters given in Table \ref{TABmktparam} and take the generator of the Markov chain to be
\begin{displaymath}
 \Markovgenerator = \begin{pmatrix} \markovgenerator_{11} & \markovgenerator_{12} \\ \markovgenerator_{21} & \markovgenerator_{22} \end{pmatrix} = \begin{pmatrix} - 0.5 & 0.5 \\ 5 & - 5 \end{pmatrix},
\end{displaymath}
These figures are based on the estimated parameters found in \citet{hardy01.article} for a $2$-state regime-switching model fitted to data from the S\&P 500, an index of 500 U.S. stocks.
 \begin{table}
 \caption{Market parameters}
  \centering
  \begin{tabular}{ | c| c | c | c | }
    \hline
    Regime $i$ & $r(i)$ & $\meanret(i)$ & $\vol (i)$ \\ \hline
     1 & 0.06 & 0.15 & 0.12 \\ \hline
     2 & 0.06 & -0.22 & 0.26 \\
    \hline
  \end{tabular}
\label{TABmktparam}
\end{table}
From the generator $\Markovgenerator$, we see that the average time spent in regime 1 is 2 years and the average time spent in regime 2 is about 2.5 months.

\subsection{Calculation and implementation}
We wish to calculate the upper and lower \gd pricing bounds for a European call option with maturity $\timehorizon=1$ and strike price $\strikeprice=100$.  To do this, we choose a \gd bound $\gooddealbnd$ and fix the initial market regime $\markovchain(0)=\markovinitialstate$ and initial stock price $\stockprice(0)$.  Then we calculate the upper and lower \gd functions of Definitions \ref{DEFNPIDEtosolveupper} and \ref{DEFNPIDEtosolvelower}.  In Section \ref{SECstochasticcontrol}, we saw that this involved first
\begin{itemize}
 \item solving the associated static optimization problem, and then
 \item numerically solving the PIDE (\ref{EQNpidetop})-(\ref{EQNpidebot}) using the solution to the static optimization problem.
\end{itemize}
We have already solved the static optimization problem for a 2-state regime-switching model, with the solution given by Lemma \ref{LMAstaticoptsoln}.  Thus it remains to numerically solve the PIDE
\begin{align}
\frac{\partial \valuefn}{\partial t} (t, x, i) + r(i) x \frac{\partial \valuefn}{\partial x} (t, x, i) + \frac{1}{2} \vol^{2} (i) x^{2} \frac{\partial^{2} \valuefn}{\partial x^{2}} (t, x, i) - r(i) \valuefn(t,x,i) & \notag \\ 
- \Pmarkovgenerator_{ii} (1 + \RNforMCopt_{ij} (t,x) ) \left( \valuefn (t, x, j) - \valuefn (t, x, i) \right) & = 0 \label{EQNMCcomponent} \\ 
\valuefn (\timehorizon, x, i)  = \max & [x-\strikeprice, 0], \notag
\end{align}
for $j \neq i$, using the optimal values $\RNforMCopt_{ij}$ which solve the static optimization problem.  Denoting the solution to the above PIDE by $\valuefnupperii$ for the upper \gd pricing bound and by $\valuefnlowerii$ for the lower \gd pricing bound, the \gd price range for the derivative is $(\valuefnlowerii(0), \valuefnupperii(0))$.  The intervals are open intervals due to the discussion in Remark \ref{RMKopenintervalgd}.

We implement the numerical solution of the PIDE using a fully implicit, finite-difference method, based on a grid which has the values
\begin{displaymath}
\Delta t = 0.01, \quad \Delta \stockprice = 0.5, \quad \stockprice_{\min} = 0, \quad \stockprice_{\max} = 200,
\end{displaymath}
where $\Delta t$ is the grid step-size in the time direction (measured in years), $\Delta \stockprice$ is the grid step-size in the stock price direction and $[\stockprice_{\min}, \stockprice_{\max}]$ is the grid range in the stock price direction.  The grid range in the time direction is $\timeset$.  We use the boundary conditions
\begin{displaymath}
 V(0,t) = 0 \qquad \textrm{and} \qquad V(\stockprice_{\max},t) = \stockprice_{\max} - \strikeprice e^{-r(\timehorizon - t)}, \quad \forall t \in \timeset.
\end{displaymath}
Note that by Definition \ref{DEFNgdb} and using the figures in Table \ref{TABmktparam} to calculate $\RNforBM (1)$ and $\RNforBM (2)$, the \gd bound $\gooddealbnd$ must satisfy
\begin{displaymath}
\gooddealbnd \geq \max [ \RNforBM^{2} (1), \RNforBM^{2} (2) ] = \max [ (-0.750)^{2}, (1.077)^{2} ] = 1.160.
\end{displaymath}

\subsection{Results} \label{SUBSECcallvarmat}
We begin by fixing the \gd bound $\gooddealbnd=1.2$ and calculating the upper and lower pricing bounds for a range of initial stock prices.  The results are shown in Figure \ref{FIGoneyrcallgdb}, with Figure \ref{FIGr1gdb12one} and \ref{FIGr2gdb12one} corresponding to the market starting in regime 1 and 2, respectively.  The middle line in each of the plots corresponds to the minimal martingale measure price, which is the benchmark price.  The plots show that, for the choice $\gooddealbnd=1.2$, the \gd pricing bounds are reasonably narrow and therefore they are potentially of practical use.
\begin{figure}[p]
\centering
\subfigure[Starting in regime 1.]
{ \label{FIGr1gdb12one}
\includegraphics[scale=0.65]{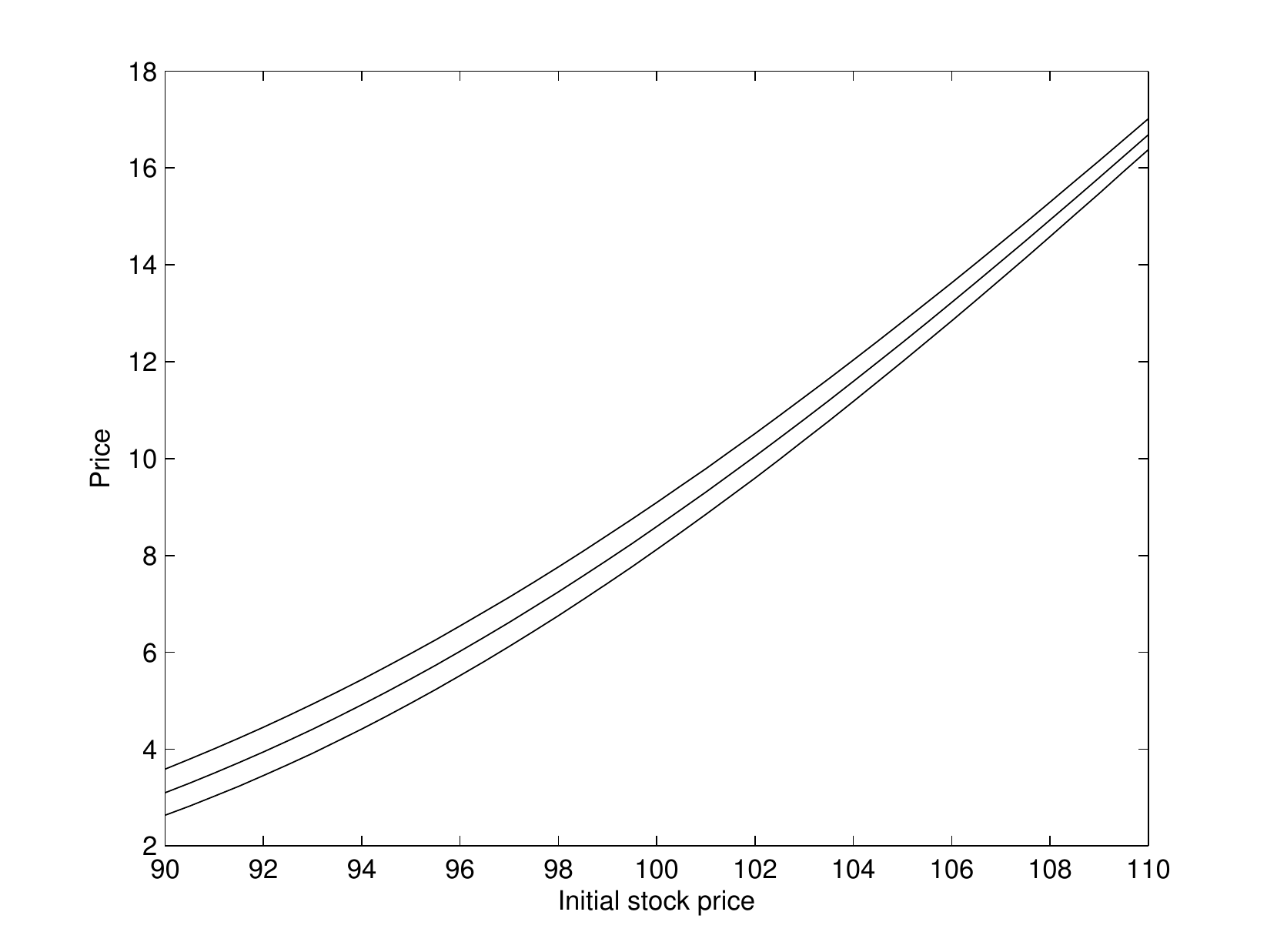}
}
\hspace{0.1cm}
\subfigure[Starting in regime 2.]
{ \label{FIGr2gdb12one}
\includegraphics[scale=0.65]{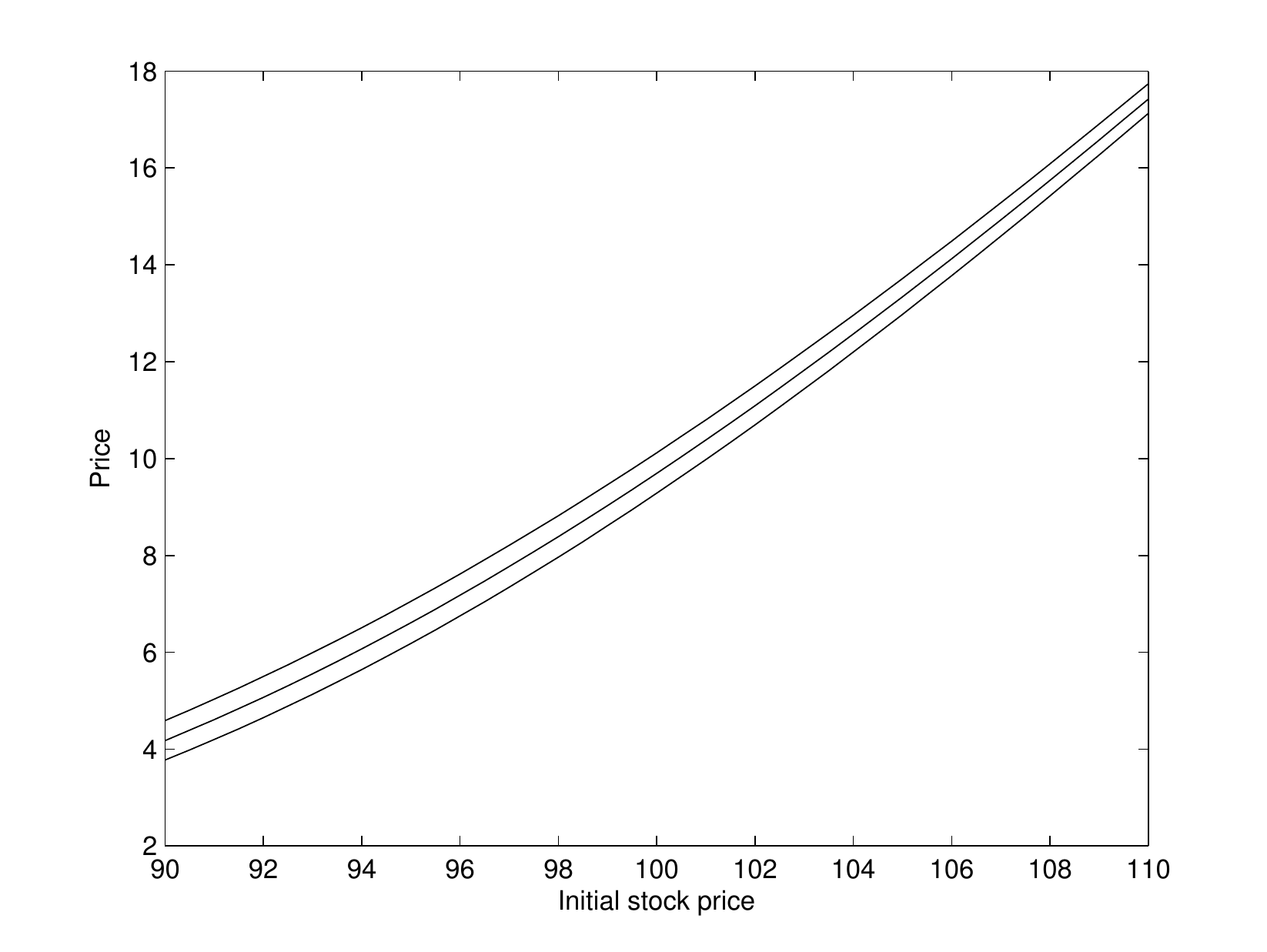}
}
\caption{The upper and lower \gd pricing bounds plotted against the initial stock price for a 1-year European call option with strike price $\strikeprice=100$ and \gd bound $\gooddealbnd=1.2$.  The upper curve on each plot is the upper \gd pricing bound, the lower curve is the lower \gd pricing bound and the middle line is the minimal martingale measure price.  The top plot assumes that the market is in regime 1 at time 0 and the bottom plot assumes that the market is in regime 2 at time 0.}
\label{FIGoneyrcallgdb}
\end{figure}
Next we examine exactly how the pricing bounds change as we vary the \gd bound $\gooddealbnd$.  We fix the initial stock price $\stockprice (0)=100$ and calculate the \gd pricing bounds for various choices of the \gd bound $\gooddealbnd$.  These results are shown in Figure \ref{FIGgdcalloption}, with Figure \ref{FIGr1S100one} and \ref{FIGr2S100one} corresponding to the market starting in regime 1 and 2, respectively.  Again, the minimal martingale measure prices are the horizontal lines in the middle of each plot.
\begin{figure}[p]
\centering
\subfigure[Starting in regime 1.]
{ \label{FIGr1S100one}
\includegraphics[scale=0.65]{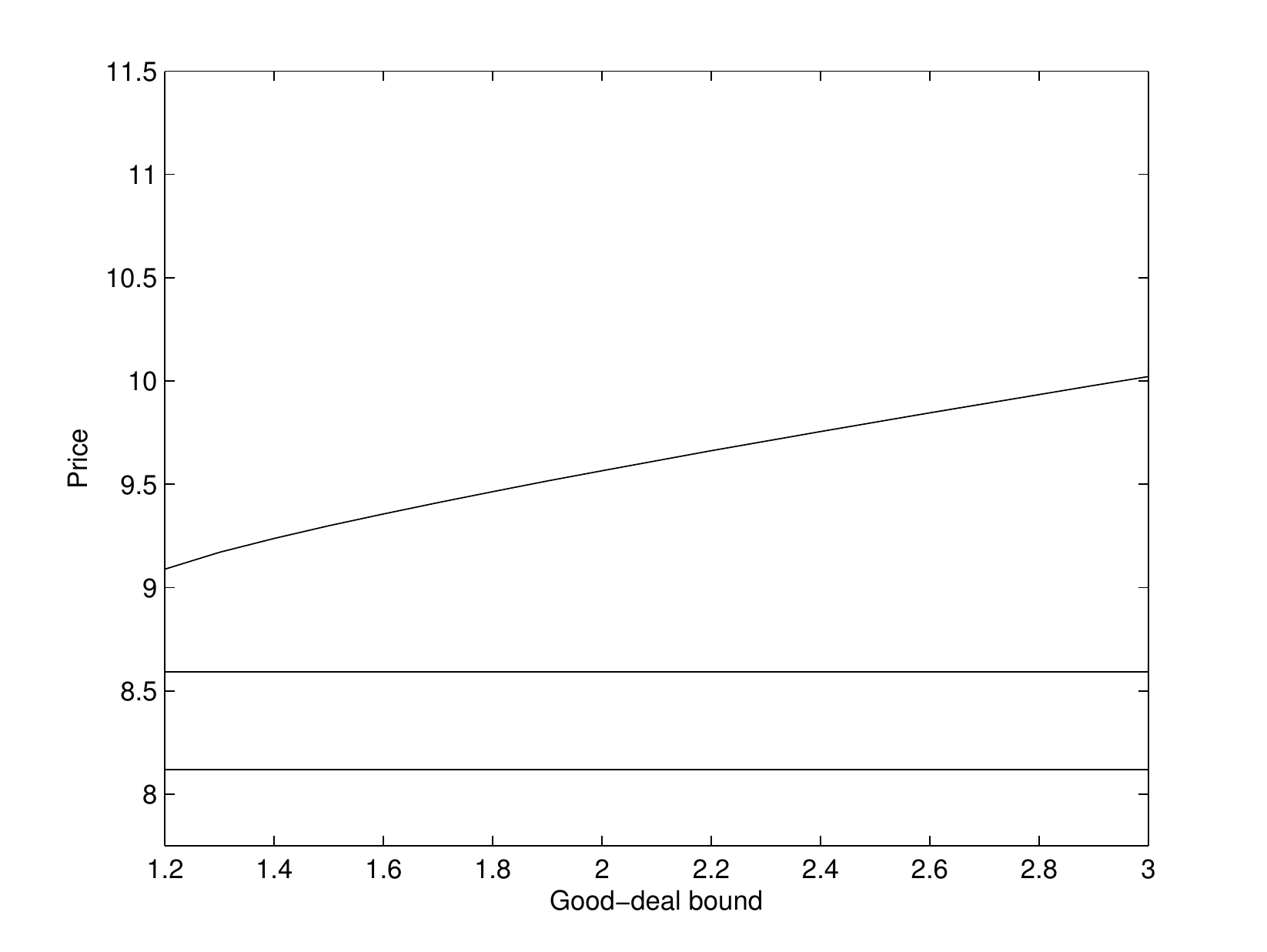}
}
\hspace{0.1cm}
\subfigure[Starting in regime 2.]
{ \label{FIGr2S100one}
\includegraphics[scale=0.65]{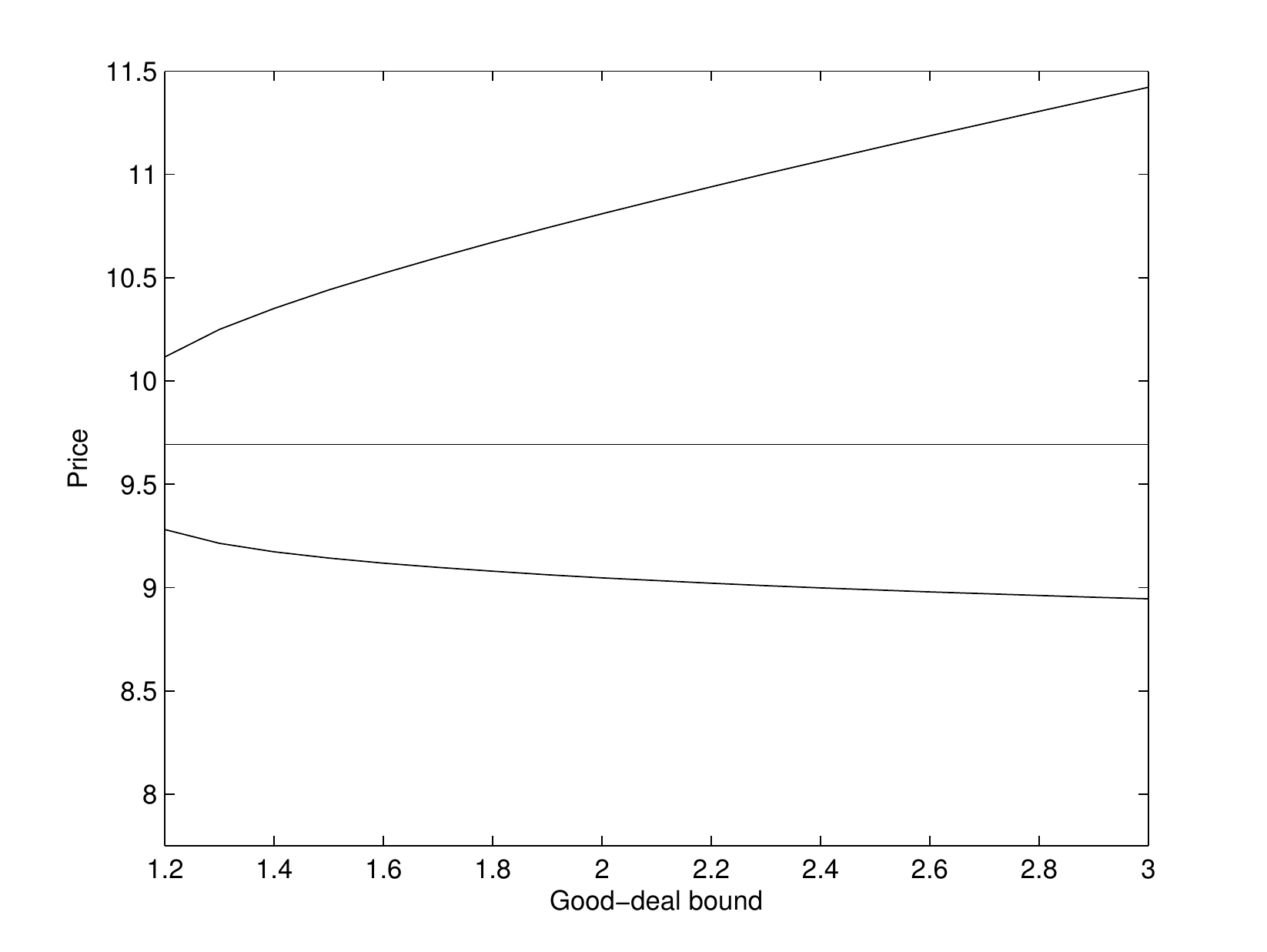}
}
\caption{The upper and lower \gd pricing bounds plotted against the \gd bound for a $1$-year European call option with strike price $\strikeprice=100$.  The initial stock price is $\stockprice (0)=100$.  The top plot assumes that the market is in regime 1 at time 0 and the bottom plot assumes that the market is in regime 2 at time 0.  On both plots, the minimal martingale measure price is the horizontal line in the middle.}
\label{FIGgdcalloption}
\end{figure}
The plots in Figure \ref{FIGgdcalloption} show that as we increase the \gd bound $\gooddealbnd$, we increase the upper \gd pricing bound.  However, while the lower \gd pricing bound decreases in Figure \ref{FIGr2S100one}, it is constant in Figure \ref{FIGr1S100one}.  The reason is that, in this particular market, the solution to the static optimization problem for the lower \gd function is always $\Qmarkovgeneratoropt_{12}^{\textrm{low}} (t, x) = -1$ when starting in regime 1, regardless of the value of the \gd bound $\gooddealbnd$.  Setting $\Qmarkovgeneratoropt_{12} (t, x) := -1$ in the PIDE (\ref{EQNMCcomponent}), we see immediately that the last term on the left-hand side vanishes and hence the PIDE reduces to the classical Black-Scholes formula for a European call option in a non-regime-switching market with market parameters $r(1)$, $b(1)$ and $\sigma(1)$.

\subsection{Stability of the \gd pricing bounds}
We base the market parameters in Table \ref{TABmktparam} on figures found in \citet{hardy01.article}.  However, the analysis in \citet{hardy01.article} gives a large standard error in the estimation of the Markov chain parameters.  This leads us to wonder what happens if we have mis-estimated the generator of the Markov chain.  Do we have stability of the \gd pricing bounds?  To examine this issue, we consider again a $1$-year European call option with strike price $\strikeprice=100$.  We find the \gd pricing bounds for this option for three different models, assuming that the price of the risky stock at time 0 is $\stockprice (0)=100$.  In each of the models, the market parameters $r$, $b$ and $\sigma$ are as in Table \ref{TABmktparam} but the diagonal elements of the generator $\Markovgenerator$ of the Markov chain are given in Table \ref{TABmodelgens}.
 \begin{table}
 \caption{Diagonal elements of the generator $\Markovgenerator$ of the Markov chain.}
  \centering
  \begin{tabular}{ | c | c | c | c |}
    \hline
     & Model 1 & Model 2 & Model 3 \\ \hline
	$-g_{11}$ & 0.5 & 0.333 & 0.667 \\
	$-g_{22}$ & 5 & 6 & 6 \\ \hline
Avg. time in regime 1 & 2.000 & 3.000 & 1.500 \\
 Avg. time in regime 2 & 0.200 & 0.167 & 0.167 \\
    \hline
  \end{tabular}
 \label{TABmodelgens}
 \end{table}
Note that Model 1 corresponds to the model described in Subsection \ref{SUBSECmodelone}.  The results are shown in Figure \ref{FIGmodelbounds}, with Figure \ref{FIGr1m123} and \ref{FIGr2m123} corresponding to the market starting in regime 1 and 2, respectively.  The middle horizontal lines in the plots correspond to the minimal martingale measure prices.

\begin{figure}[p]
\centering
\subfigure[Starting in regime 1.]
{ \label{FIGr1m123}
\includegraphics[scale=0.65]{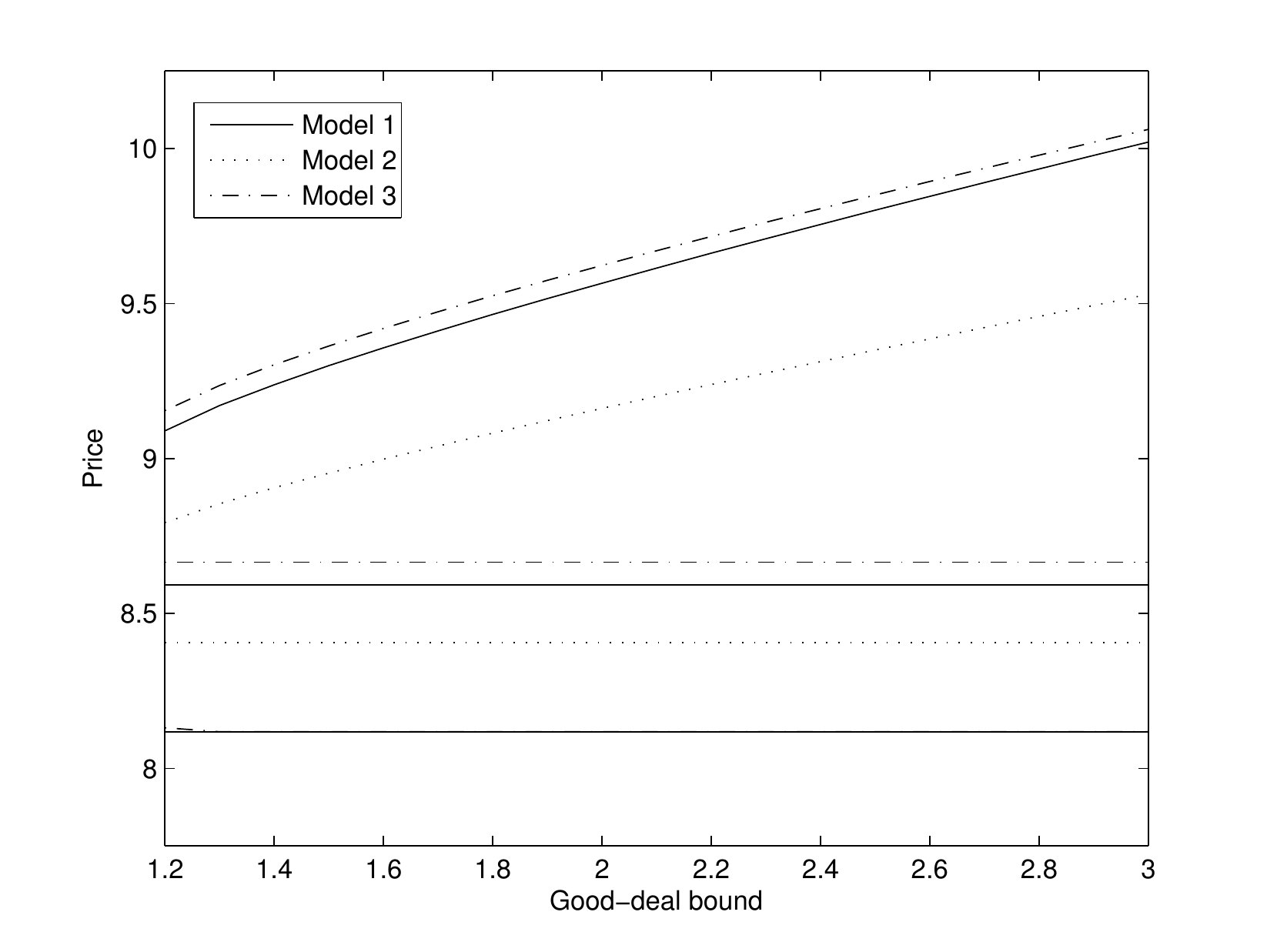}
}
\hspace{0.1cm}
\subfigure[Starting in regime 2.]
{ \label{FIGr2m123}
\includegraphics[scale=0.65]{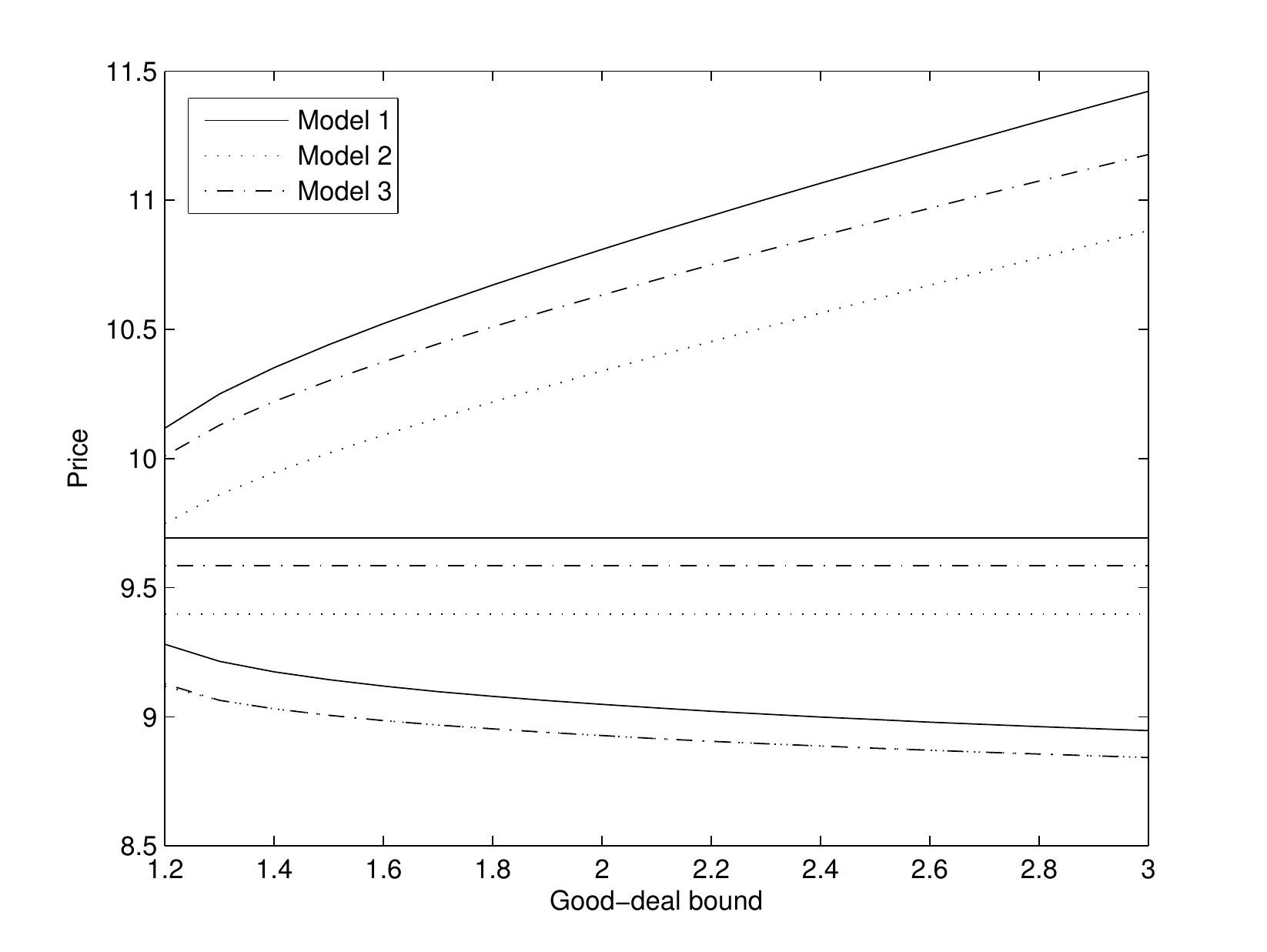}
}
\caption{The upper and lower \gd pricing bounds plotted against the \gd bound for a $1$-year European call option with strike price $\strikeprice:=100$.  The bounds are calculated for three market models with identical market parameters (given in Table \ref{TABmktparam}) but with different generators $\Markovgenerator$ of the Markov chain (with diagonal elements given in Table \ref{TABmodelgens}).  The top plot assumes that the market is in regime 1 at time 0 and the bottom plot assumes that the market is in regime 2 at time 0.  For each model, the minimal martingale measure price is the horizontal middle line.}
\label{FIGmodelbounds}
\end{figure}
The plots show that the \gd pricing bounds are sensitive to the choice of the generator $\Markovgenerator$ of the Markov chain.  Roughly, the upper \gd pricing bounds move in tandem with the minimal martingale measure prices.  In Figure \ref{FIGr2m123}, the lower \gd pricing bounds behave similarly.  However, in Figure \ref{FIGr1m123}, the lower \gd pricing bounds are nearly all constant.  The explanation for the constant lower pricing bounds is as before: the solution to the static optimization problem for the lower \gd function is $\Qmarkovgeneratoropt_{12}^{\textrm{low}} (t, x) = -1$ when starting in regime 1 in these cases.  In Figure \ref{FIGr1m123} we see that lower \gd pricing bound is slightly higher at the \gd bound $\gooddealbnd=1.2$ because here the solution $\Qmarkovgeneratoropt_{12}^{\textrm{low}} (t, x)$ to the static optimization problem for the lower \gd function is just above $-1$.

\section{Conclusion} \label{SECconclude}
We have applied the \gd bound idea of \cite{cochranesaa00.article} to a regime-switching market using the approach of \citet{bjorkslinko06.article} and illustrated it with a numerical example.  The \gd bound idea is a way to measure the uncertainty in the choice of the risk-neutral measure used to price derivatives.  However, as our numerical example demonstrates, the \gd pricing bounds change when the model changes, even though the \gd bound remains constant.  Thus the \gd pricing bounds are sensitive to the choice of model.  It would be interesting to do a wider investigation of the variation of the \gd pricing bounds over a wider class of models for different derivatives.

We have looked at pricing, but what is a ``\gd hedging'' strategy?   As \citet{bjorkslinko06.article} say in their conclusion, this is a highly challenging open problem.

\section*{Acknowledgments}
The author thanks RiskLab, ETH Zurich, Switzerland for financial support.  The author is grateful to Tomas Bj{\"o}rk for his sound advice and his correction of some errors in an early draft of this paper.  Paul Embrechts, Marius Hofert and an anonymous referee made valuable comments on the presentation of the paper.  The author thanks Michel Baes for a helpful discussion.

\bibliographystyle{plainnat}
\bibliography{article}

\end{document}